\documentclass[11p,reqno]{amsart}

\usepackage[T1]{fontenc}
\usepackage{graphicx}
\usepackage{amssymb,amsthm,amsmath,mathrsfs,bm,braket,marginnote}
\usepackage{enumerate}
\usepackage{appendix}
\usepackage[colorlinks=true, pdfstartview=FitV, linkcolor=blue, citecolor=blue, urlcolor=blue]{hyperref}
\usepackage{multirow}

\usepackage{pgf}
\usepackage{pgfplots}
\usepackage{tikz}
\usetikzlibrary{arrows,calc}
\usepackage{verbatim}
\usetikzlibrary{decorations.pathreplacing,decorations.pathmorphing}
\usepackage[numbers,sort&compress]{natbib}
\usepackage{dsfont}

\numberwithin{equation}{section}
\newtheorem{theorem}{Theorem}[section]

\newtheorem{definition}[theorem]{Definition}
\newtheorem{remark}[theorem]{Remark}
\newtheorem{coro}[theorem]{Corollary}

\newtheorem{proposition}[theorem]{Proposition}

 \reversemarginpar
 \newcommand{\pf}{\text{pf}}
 \newcommand{\Pf}{\text{Pf}}

\newcommand{\s}{\mathbb{S}}

\newcommand{\p}{\partial}

\newcommand{\mc}{\mathcal{C}}

\newcommand{\mt}{\mathbf{t}}
\newcommand{\ms}{\mathbf{s}}

\begin{document}
	
\title[Multiple skew orthogonal polynomials and 2-component Pfaff lattice hierarchy]{Multiple skew orthogonal polynomials and 2-component Pfaff lattice hierarchy}

\author{Shi-Hao Li}
\address{Department of Mathematics, Sichuan University, Chengdu, 610064, China}
\email{lishihao@lsec.cc.ac.cn}

\author{Bo-Jian Shen}
\address{School of Mathematical Sciences, CMA-Shanghai, Shanghai Jiaotong University, People's Republic of China}
\email{john-einstein@sjtu.edu.cn}

\author{Jie Xiang}
\address{Department of Mathematics, Sichuan University, Chengdu, 610064, China}
\email{oreki$\_$phy@163.com}

\author{Guo-Fu Yu}
\address{School of Mathematical Sciences, CMA-Shanghai, Shanghai Jiaotong University, People's Republic of China.}
\email{gfyu@sjtu.edu.cn}

\date{}

\dedicatory{}
\subjclass[2020]{37K10, 15A23}
\keywords{Multiple skew orthogonal polynomials; Pfaff $\tau$-function; 2-component Pfaff hierarchy}

\begin{abstract}
In this paper, we introduce multiple skew-orthogonal polynomials and investigate their connections with classical integrable systems. By using Pfaffian techniques, we show that multiple skew-orthogonal polynomials can be expressed by multi-component Pfaffian tau-functions upon appropriate deformations. Moreover, a two-component Pfaff lattice hierarchy, which is equivalent to the Pfaff-Toda hierarchy studied by Takasaki, is obtained by considering the recurrence relations and Cauchy transforms of multiple skew-orthogonal polynomials.
\end{abstract}

\maketitle

\section{Introduction}

In recent decades, the interplay between random matrix theory and integrable systems attracted much attention due to the development of both fields. A crucial observation in the connection is that the partition functions of different random matrix models can act as the $\tau$-functions of corresponding integrable hierarchies upon appropriate deformations. Such an observation was found by making use of semi-classical orthogonal polynomials for different integrable hierarchies, such as Painlev\'e hierarchy \cite{bertola06,ormerod11} and Toda hierarchy \cite{adler95,aptekarev97}. 

It is well known that a sequence of orthogonal polynomials $\{p_n(x)\}_{n\in\mathbb{N}}$ can be characterized by an analytic, non-negative weight $\omega(x)$ such that
\begin{align}\label{or}
\int_{\mathbb{R}} p_n(x)p_m(x)\omega(x)dx=\delta_{n,m}.
\end{align}
According to Favard's theorem \cite{favard35,chihara78}, the orthogonal relation \eqref{or} can be equivalently expressed by a three-term recurrence relation
\begin{align}\label{ttrr}
xp_n(x)=a_np_{n+1}(x)+b_np_n(x)+a_{n-1}p_{n-1}(x),\quad p_{-1}(x)=0,\quad p_0(x)=1,
\end{align}
for a sequence of coefficients $\{a_n,b_n\}_{n\in\mathbb{N}}$, providing a Jacobi matrix form
\begin{align*}
L=\left(\begin{array}{ccccc}
b_0&a_0&&&\\
a_0&b_1&a_1&&\\
&a_1&b_2&a_2&\\
&&\ddots&\ddots&\ddots
\end{array}
\right),\quad x\Psi(x)=L\Psi(x),\quad \Psi(x)=(p_0(x),p_1(x),\cdots)^\top.
\end{align*}

Semi-classical orthogonal polynomials were firstly considered by Shobat \cite{shobat39} and later by Freud \cite{freud76}. In such case, time parameters $\mt=(t_1,t_2,\cdots)$ were introduced into the weight such that 
\begin{align*}
\p_{t_n}\omega(x;\mt)=x^n\omega(x;\mt).
\end{align*}
Therefore, orthogonal polynomials with semi-classical weight are time-dependent and result in the formula
\begin{align}\label{timee}
\p_{t_1}p_n(x;\mt)=-\frac{1}{2}b_np_n(x;\mt)-a_np_{n-1}(x;\mt).
\end{align}
In literatures, there are several ways to derive integrable lattices from semi-classical orthogonal polynomials. One is a direct method by using the compatibility condition of \eqref{ttrr} and \eqref{timee}, from which one gets
\begin{align*}
\p_{t_1}a_n=\frac{1}{2}a_n(b_n-b_{n-1}),\quad \p_{t_1}b_n=a_{n-1}^2-a_n^2,
\end{align*}
and this is the nonlinear form for the Toda lattice. Details and related discussions can be found in monographs \cite{forrester10,deift00}. Another way is to express orthogonal polynomials by $\tau$-functions, and integrable hierarchies can be obtained by the recurrence relation. It is known that by solving the orthogonal relation \eqref{or}, a determinantal expression for $p_n(x;\mt)$ is given by
\begin{align*}
p_n(x;\mt)=\frac{1}{\sqrt{\tau_n(\mt)\tau_{n+1}(\mt)}}\det\left(\begin{array}{cccc}
m_0&m_1&\cdots&m_n\\
\vdots&\vdots&&\vdots\\
m_{n-1}&m_n&\cdots&m_{2n-1}\\
1&x&\cdots&x^n
\end{array}
\right),
\end{align*}
where 
\begin{align*}
\tau_n(\mt)=\det(m_{i+j})_{i,j=0}^{n-1},\quad m_i=\int_{\mathbb{R}}x^i\omega(x;\mt)dx.
\end{align*}
Shifting $\mt$ backwards by $[x^{-1}]$ in the $\tau$-function yields a polynomial in $x$, we have
\begin{align}\label{taue}
p_n(x;\mt)=x^n\frac{\tau_n(\mt-[x^{-1}])}{\sqrt{\tau_n(\mt)\tau_{n+1}(\mt)}},\quad [x^{-1}]=\left(
\frac{x^{-1}}{1},\frac{x^{-2}}{2},\cdots
\right).
\end{align}
Moreover, if one substitutes such formula into the recurrence relation \eqref{ttrr}, the Toda hierarchy with neighboring points is obtained. Such hierarchy can also be derived by using Cauchy transforms. If one considers the Cauchy transform of orthogonal polynomials
\begin{align*}
\int_{\mathbb{R}}\frac{p_n(x;\mt)}{z-x}\omega(x;\mt)dx=z^{-n-1}\frac{\tau_{n+1}(\mt+[z^{-1}])}{\sqrt{\tau_n(\mt)\tau_{n+1}(\mt)}},
\end{align*}
then from the orthogonality, one has the formula
\begin{align}
\begin{aligned}
0=\int_{\mathbb{R}}p_n(x;\mt)p_{n-1}(x;\mt')\omega(x;t)dx=\frac{1}{2\pi i}\oint_{C_\infty} \tau_n(\mt-[z^{-1}])\tau_{n}(\mt'+[z^{-1}])e^{\xi(\mt,z)-\xi(\mt',z)}dz,
\end{aligned}
\end{align}
where $\xi(\mt,z)=\sum_{i=1}^\infty t_iz^i$. This formula is valid for all $t,\,t'\in\mathbb{C}$ and gives a bilinear identity of KP hierarchy \cite{adler95,jimbo83}. 

Relations between orthogonal polynomials and integrable systems are clearly depicted by considering different  generalizations of the orthogonal relation \eqref{or},
which, in fact, is given by a symmetric, positive definite, and real bilinear form
\begin{align*}
\langle \cdot,\cdot\rangle:\mathbb{R}[x]\times\mathbb{R}[x]\to\mathbb{R}
\end{align*}
such that $\langle x^i,x^j\rangle=\langle x^j,x^i\rangle$. Therefore, the generalizations of orthogonality is equivalent to the extensions of the bilinear form.  A non-symmetric generalization to the bilinear form admits
\begin{align*}
\langle x^i,x^j\rangle=\int_{\mathbb{R}}x^{i+\theta j}\omega(x)dx,\quad \theta\in\mathbb{R}_+.
\end{align*}
This bilinear form is related to the random matrix models with additional interaction proposed by Muttalib and Borodin, and corresponding polynomials were referred to as bi-orthogonal polynomials \cite{muttalib95,borodin98}. There is another kind of bi-orthogonality by considering a bilinear form acting on $\mathbb{R}[x]\times\mathbb{R}[y]$, such that
\begin{align}\label{bf1}
\langle x^i,y^j\rangle=\int_{\mathbb{R}^2} x^iy^j\mathbb{K}(x,y)\omega_1(x)\omega_2(y)dxdy,
\end{align}
where $\mathbb{K}(x,y)$ is a kernel function and $\omega_1$, $\omega_2$ are weights with respect to $x$ and $y$ respectively. Such bi-orthogonal polynomials were introduced by considering matrices coupled in a chain \cite{eynard98} and Cauchy two-matrix models \cite{bertola09}. Specifically, skew-symmetric kernels arisen from orthogonal and symplectic ensembles in random matrix models are of particular interest. Above mentioned orthogonal polynomials are all related to integrable systems if appropriate time deformations are assumed. Examples include Gelfand-Dickey hierarchy (Muttalib-Borodin case) \cite{tsujimoto00}, 2d-Toda hierarchy (coupled chain case) \cite{adler99},  CKP hierarchy (Cauchy two-matrix model case) \cite{li19}, Pfaff lattice/DKP hierarchy (orthogonal/symplectic ensemble case) \cite{adler992,kakei00} and BKP hierarchy (Bures ensemble case) \cite{hu17}. 

Multiple orthogonal polynomials (MOPs) as a generalization of orthogonal polynomials is a sequence of polynomials orthogonal with several different weights originated in the study of what is termed Hermite--Pad\'e approximation. This is the simultaneous rational approximation
of a family of functions $\{ f_j \}$ which allow for a decaying Laurent expansion at infinity. Such functions can be written as
\begin{equation}\label{1.1}
f_j(z) = \int_{I_j} {d \mu_j(x) \over z - x},
\end{equation}
for several measures $\{\mu_j\}$. It is these measures which directly relate to the orthogonality of MOPs; see e.g.~the brief survey \cite{mf16}.

A relatively recent application of MOPs is in the field of random matrices. The Gaussian unitary ensemble is the set of
$N \times N$ random complex Hermitian matrices $\{H\}$, chosen with a probability density function (PDF) proportional to
$e^{- {\rm Tr} \, H^2}$. In particular, the diagonal entries are all independent real normal random variables with
mean zero and standard deviation $1/\sqrt{2}$ (denoted N$[0,1/\sqrt{2}]$),
while the upper triangular entries of $H$ are similarly independent and identically distributed, with complex
normal distribution
N$[0,1/2] + i {\rm N}[0,1/2]$.  Modifying this ensemble so that the entries have a non-zero mean, the corresponding
PDF becomes proportional to $e^{- {\rm Tr} \, (H - A)^2}$, where $A$ is a fixed complex Hermitian matrix. 
The new ensemble is referred to as the Gaussian unitary ensemble with a source \cite{brezin98}.
Let the
eigenvalues of $A$ be denoted $\{a_j\}$. A result of Bleher and Kuijlaars \cite{bleher04} gives that the average
characteristic polynomial $\langle x \mathbb I - (H - A) \rangle$ can be expressed in terms of a particular type
II MOPs\footnote{For formal definitions of type I and type II MOPs, please refer to Sec \ref{sec2.1}.}--- referred to as multiple Hermite polynomials --- where the family of measures are proportional to $\{ e^{-x^2 + 2a_j x}\}_{j=1}^N$. This same random
matrix model, and thus the relevance of the multiple Hermite polynomials, relates to non-intersecting Brownian bridges \cite{aptekarev05}. 
Moreover, in \cite{desrosiers08} the chiral generalization of the Gaussian unitary ensemble with a source is related to particular type I and type II Laguerre MOPs. 
With type I and type II MOPs closely related to non-intersecting Brownian motions, a generalized MOPs called  mixed type MOPs was proposed in \cite{daems07} to make further assumptions on paths, and their applications into integrable systems were considered in \cite{adler09,af11,aptekarev16}.

In this paper, we focus on a generalization of skew-orthogonal polynomials called multiple skew-orthogonal polynomials (MSOPs) and make connections with integrable hierarchies. Skew-orthogonal polynomials arise when the integral kernel in \eqref{bf1} is assumed to be skew-symmetric. Therefore, to give a proper definition of MSOPs, we firstly consider a bi-orthogonal generalization of MOPs in Section \ref{sec2.2}. Symmetric and skew symmetric reductions are considered in Section \ref{sec2.3} to give a determinant expressions for MSOPs. Section \ref{sec3} is devoted to the 2-component MSOPs, which are skew orthogonal with weights $\omega_1$ and $\omega_2$. Proposition \ref{prop3.1} states that 2-component MSOPs admit Pfaffian expressions as well, from which 2-component Pfaffian $\tau$-functions could be involved. Then we introduce two different sets of time variables $\mt=(t_1,t_2,\cdots)$ and $\ms=(s_1,s_2,\cdots)$ into weights $\omega_1$ and $\omega_2$ respectively, and prove deformation identities by making use of Pfaffian notations. Such identities are helpful in deriving integrable systems. Analogous to the standard orthogonal polynomials and Toda lattice hierarchy, we apply three different methods to derive integrable lattices ranging from simple to complex. The first one is shown in Section \ref{sec3} by simply comparing the coefficients in the deformation identities, and several simple equations are demonstrated. Furthermore, a systematic study in the derivation of integrable lattice hierarchy is carried out in Section \ref{sec4} from two different perspectives. One is to show that the above mentioned Pfaffian expressions can be alternatively expressed by $\tau$-functions with time evolutions. By substituting $\tau$-functions expressions into identities satisfied by MSOPs, one gets an integrable hierarchy for neighboring $\tau$-functions. A shortage in this strategy is that only neighboring $\tau$-functions are involved in resulting integrable hierarchy. We promote this method by considering a Cauchy transform method. In Section \ref{sec4.2}, we utilize the Cauchy transform of MSOPs and show that Takasaki's Pfaff-Toda hierarchy is equivalent to our 2-component Pfaff lattice hierarchy.

\section{Multiple skew-orthogonal polynomials}
In this part, we intend to introduce the concept of multiple skew-orthogonal polynomials, which are skew-orthogonal with respect to several different weights. Multiple skew-orthogonality is originated from the multiple orthogonality, and thus a brief review of the latter is firstly given to make the paper self-consistent.

\subsection{A brief review of MOPs}\label{sec2.1}
Multiple orthogonal polynomials (MOPs) are defined as polynomials of one variable that satisfy orthogonality conditions with respect to several weights \cite[Chap. 23]{ismail09}. Given a multi-index $\vec{v}\in \mathbb{N}^{p}$ with length $|\vec{v}|=\sum_{i=1}^p v_i$, and $p$ different weight functions $(\omega_1,\cdots,\omega_p)$ supported on the real line, there are two types of MOPs. Type I MOPs are collected in a vector of $ p $ polynomials 
$(A_{\vec{v},1},\cdots,A_{\vec{v},p})$, where each $ A_{\vec{v},i} $ has degree at most $ v_i-1 $, satisfying the orthogonality relations
\begin{align}\label{t1mop}
\int_{\mathbb{R}} x^k\left(
\sum_{i=1}^p A_{\vec{v},i}(x)\omega_i(x)
\right)dx=\delta_{k,|\vec{v}|-1},\quad 0\leq k\leq |\vec{v}|-1.
\end{align}
By assuming $$A_{\vec{v},i}(x)=\xi_{i,v_i-1}x^{v_i-1}+\cdots+\xi_{i,0},$$
the above relations give rise to a linear system of $ |\vec{v}| $ equations for $ |\vec{v}| $ unknown coefficients $\{\xi_{i,j},j=0,\cdots,v_i-1,i=1,\cdots,p\}$
\begin{align}\label{type1}
\left(\begin{array}{ccccccc}
m_0^{(1)}&\cdots&m_{v_1-1}^{(1)}&\cdots&m_{0}^{(p)}&\cdots&m_{v_p-1}^{(p)}\\
\vdots&&\vdots&&\vdots&&\vdots\\
m_{v_1-1}^{(1)}&\cdots&m_{2v_1-2}^{(1)}&\cdots&m_{v_1-1}^{(p)}&\cdots&m_{v_1+v_p-2}^{(p)}\\
\vdots&&\vdots&&\vdots&&\vdots\\
m_{|\vec{v}|-v_p}^{(1)}&\cdots&m_{|\vec{v}|+v_1-v_p-1}^{(1)}&\cdots&m_{|\vec{v}|-v_p}^{(p)}&\cdots&m_{|\vec{v}|-1}^{(p)}\\
\vdots&&\vdots&&\vdots&&\vdots\\
m_{|\vec{v}|-1}^{(1)}&\cdots&m_{|\vec{v}|+v_1-2}^{(1)}&\cdots&m_{|\vec{v}|-1}^{(p)}&\cdots&m_{|\vec{v}|+v_p-2}^{(p)}\end{array}
\right)\left(\begin{array}{c}
\xi_{1,0}\\\vdots\\\xi_{1,v_1-1}\\\vdots\\\xi_{p,0}\\\vdots\\\xi_{p,v_p-1}
\end{array}
\right)=\left(\begin{array}{c}
0\\\vdots\\0\\\vdots\\0\\\vdots\\1
\end{array}
\right),
\end{align}
where moments are defined by $m_j^{(i)}=\int_\mathbb{R} x^j \omega_i(x)dx$. 

The polynomials $\{ A_{\vec{v},i}, i=1,\dots,p \}$ are uniquely determined if and only if the linear system has a unique solution, which requires the determinants of moment matrices to be nonzero. This condition gives restrictions on the weights $ \omega_1,\cdots,\omega_p $. In general, there is no guarantee that for a given multi-index, the corresponding MOPs exist. A multi-index $\vec{v}$ is said to be normal for type I MOPs if $\{ A_{\vec{v},i}, i=1,\dots,p \}$ exists and is unique. If all multi-indices are normal, then the system of weights $(\omega_1,\cdots,\omega_p)$ is said to be a perfect system. There are two well-known perfect systems, one is the Angelesco system and the other is the Nikishin system, where the perfectness of the former is given by the properties of zeros of orthogonal polynomials, and that of the latter is due to the analytic property of weights. For details, please refer to \cite{fidalgo11,nikishin91}.

By considering the dual construction, type II MOPs $\{P_{\vec{v}}(x)\}$ are defined as scalar polynomials with degree $|\vec{v}|$ by the orthogonal relation
\begin{align}\label{t2}
\int_\mathbb{R} P_{\vec{v}}(x) x^j \omega_{i}(x)dx=0,\quad j=0,\cdots,v_i-1,\quad i=1,\cdots, p.
\end{align}
If we assume $P_{\vec{v}}(x)$ to be monic as a normalized condition, then a linear system of $ |\vec{v}| $ equations is read from orthogonal relations. By assuming that $ P_{\vec{v}}(x) =x^{|\vec{v}|}+\eta_{|\vec{v}|,|\vec{v}|-1}x^{|\vec{v}|-1}+\cdots+\eta_{|\vec{v}|,0}$, we have
\begin{align}\label{type2}
\left(\begin{array}{ccc}
m_0^{(1)}&\cdots&m_{|\vec{v}|-1}^{(1)}\\
\vdots&&\vdots\\
m_{v_1-1}^{(1)}&\cdots&m_{|\vec{v}|+v_1-2}^{(1)}\\
\vdots&&\vdots\\
m_0^{(p)}&\cdots&m_{|\vec{v}|-1}^{(p)}\\
\vdots&&\vdots\\
m_{v_p-1}^{(p)}&\cdots&m_{|\vec{v}|+v_p-2}^{(p)}\end{array}
\right)\left(\begin{array}{c}
\eta_{|\vec{v}|,0}\\
\vdots\\
\eta_{|\vec{v}|,v_1-1}\\
\vdots\\
\eta_{|\vec{v}|,|\vec{v}|-v_p+1}\\
\vdots\\
\eta_{|\vec{v}|,|\vec{v}|-1}\end{array}
\right)
=-\left(\begin{array}{c}
m_{|\vec{v}|}^{(1)}\\
\vdots\\
m_{|\vec{v}|+v_1-1}^{(1)}\\
\vdots\\
m_{|\vec{v}|}^{(p)}\\
\vdots\\
m_{|\vec{v}|+v_p-1}^{(p)}\end{array}
\right),\quad m_j^{(i)}=\int_\mathbb{R} x^j\omega_i(x)dx.
\end{align}
Similar to the type I case, we say that $ \vec{v} $ is a normal index for type II MOPs if the linear system has a unique solution. By noting that the coefficient matrix in \eqref{type2} is the transpose of that for type I in \eqref{type1}, we know that a multi-index is normal for type II if and only if it is normal for type I.

Moreover, if $u=(u_1,\cdots,u_{p_1})$ and $v=(v_1,\cdots,v_{p_2})$ are two multi-indices, and $\vec{\omega}=(\omega_1,\cdots,\omega_{p_1})$ is a set of weights, let's
define type I function 
\begin{align*}
Q_{\vec{u}}(x)=\sum_{i=1}^{p_1}A_{\vec{u},i}(x)\omega_i(x),
\end{align*}
and type II MOP $P_{\vec{v}}(x)$ with regard to weight $\vec{\omega}$.
Then there is a bi-orthogonality property \cite[Thm. 23.1.6]{ismail09}
\begin{align}\label{mopor}
\int_{\mathbb{R}}P_{\vec{v}}(x)Q_{\vec{u}}(x)dx=\left\{\begin{array}{ll}
0& \text{if $\vec{u}\leq \vec{v}$,}\\
0& \text{if $|\vec{v}|\leq |\vec{u}|-2$,}\\
1& \text{if $|\vec{v}|=|\vec{u}|-1$.}\\
\end{array}
\right.
\end{align}

Except for type I and type II MOPs, a family of mixed MOPs was proposed in the study of non-intersecting Brownian motions \cite{daems07}. Let's consider a non-intersecting Brownian motion on $\mathbb{R}$, with $u_{\alpha}$ paths starting at $a_\alpha\in\mathbb{R}$ $(\alpha=1,\cdots,p_1)$, and with $v_{\beta}$ paths ending at points $b_\beta\in\mathbb{R}$ $(\beta=1,\cdots,p_2)$. Since there is no collision in paths, we require 
\begin{align}\label{length}
\sum_{\alpha=1}^{p_1}u_{\alpha}=\sum_{\beta=1}^{p_2}v_{\beta}.
\end{align} This equation plays an important role in the definition of mixed MOPs and will be explained later.
 Applications of mixed MOPs in recent years were proposed in integrable system and  random walks \cite{daems07,adler09,af11,branquinho21}. What is special about mixed MOPs is that they are orthogonal with two different sets of weights. Assume that $\vec{u}=(u_{1},\cdots,u_{p_1})$ and $\vec{v}=(v_{1},\cdots,v_{p_2})$ are two multi-indices, and $\vec{\omega}_{1}=(\omega_{1,1},\cdots,\omega_{1,p_1})$ and $\vec{\omega}_{2}=(\omega_{2,1},\cdots,\omega_{2,p_2})$ are two sets of weights, then a family of  polynomials $A_{1},\cdots,A_{p_1}$ with deg $A_i\leq u_i-1$ could be defined by orthogonal relations 
\begin{align}\label{or-mt}
\int_{\mathbb{R}} \left(
\sum_{i=1}^{p_1}A_{i}(x)\omega_{1,i}(x)
\right)\omega_{2,j}(x)x^kdx=0,\quad k=0,\cdots,v_{j}-1,\quad j=1,\cdots,p_2.
\end{align}
Polynomials $ A_{1},\cdots,A_{p_1} $ are called MOPs of mixed type since the function 
\begin{align*}
P_{\vec{u},\vec{v}}(x)=\sum_{i=1}^{p_1}A_{i}(x)\omega_{1,i}(x)
\end{align*}
is a linear form of the first set of weights as in type I multiple orthogonality (c.f. \eqref{t1mop}) and has the same type of orthogonality with respect to the second set of weights as in type II multiple orthogonality (c.f. \eqref{t2}). Given another pair of indices $ \vec{u}'=(u_1',\dots,u_{p_1}') $ and $ \vec{v}'=(v_1',\dots,v_{p_2}') $, one can also consider a family of polynomials $ B_1,\dots,B_{p_2} $ with deg $ B_i\leq v_i'-1 $ such that the linear form 
\begin{align*}
Q_{\vec{u}',\vec{v}'}(x)=\sum_{i=1}^{p_2}B_{i}(x)\omega_{2,i}(x)
\end{align*}
satisfy the orthogonal relations
\begin{align}\label{or-mt2}
\int_{\mathbb{R}} x^k\omega_{1,j}(x)Q_{\vec{u}',\vec{v}'}(x)
dx=0,\quad k=0,\cdots,u_{j}'-1,\quad j=1,\cdots,p_1.
\end{align}
As a simple observation, the orthogonality \eqref{or-mt} and \eqref{or-mt2} can be established equivalently by the formula
\begin{align}\label{orePQ}
\int_{\mathbb{R}}P_{\vec{u},\vec{v}}(x)Q_{\vec{u}',\vec{v}'}(x)dx=0\text{ for $\vec{u}\leq\vec{u}'$ or $\vec{v}\geq \vec{v}'$}.
\end{align}
The partial order relation $ \vec{u}\leq\vec{u}' $ means that $ u_i\leq u_i' $ for every $ i \in [1,p_1]$.
If we denote moments $$m_j^{(l,k)}=\int_\mathbb{R} x^j \omega_{1,l}(x)\omega_{2,k}(x)dx$$ and assume that $A_i(x)=\xi_{i,u_{i}-1}x^{u_{i}-1}+\cdots+\xi_{i,0}$, then orthogonal conditions \eqref{or-mt} results in the following linear system
\begin{align*}
\left(\begin{array}{ccccccc}
m_0^{(1,1)}&\cdots&m_{u_{1}-1}^{(1,1)}&\cdots&m_0^{(p_1,1)}&\cdots&m_{u_{p_1-1}}^{(p_1,1)}\\
\vdots&&\vdots&&\vdots&&\vdots\\
m_{v_{1}-1}^{(1,1)}&\cdots&m_{u_{1}+v_{1}-2}^{(1,1)}&\cdots&m_{v_{1}-1}^{(p_1,1)}&\cdots&m_{u_{p_1}+v_{1}-2}^{(p_1,1)}\\
\vdots&&\vdots&&\vdots&&\vdots\\
m_0^{(1,p_2)}&\cdots&m_{u_{1}-1}^{(1,p_2)}&\cdots&m_{0}^{(p_1,p_2)}&\cdots&m_{u_{p_1}-1}^{(p_1,p_2)}\\
\vdots&&\vdots&&\vdots&&\vdots\\
m_{v_{p_2}-1}^{(1,p_2)}&\cdots&m_{u_{1}+v_{p_2}-2}^{(1,p_2)}&\cdots&m_{v_{p_2}-1}^{(p_1,p_2)}&\cdots&m_{u_{p_1}+v_{p_2}-2}^{(p_1,p_2)}\end{array}
\right)\left(\begin{array}{c}
\xi_{1,0}\\\vdots\\\xi_{1,u_{1}-1}\\\vdots\\\xi_{p_1,0}\\\vdots\\\xi_{p_1,u_{p_1}-1}
\end{array}
\right)=0
\end{align*}
with $|\vec{v}|$ equations and $|\vec{u}|$ unknowns. Therefore, to ensure a nonzero solution of the linear system, one needs to assume that $|\vec{u}|=|\vec{v}|+1$ (for $ Q_{\vec{u}',\vec{v}'} $, we require $ |\vec{u}'|+1=|\vec{v}'| $). By solving the linear equations directly using the Cramer's rule, we see that the linear form $ P_{\vec{u},\vec{v}}(x) $ and $ Q_{\vec{u}',\vec{v}'}(x) $ are proportional to determinants
\begin{align*}
&P_{\vec{u},\vec{v}}(x)=\sum_{i=1}^{p_1}A_i(x)\omega_{1,i}(x)\sim \det\left(\begin{array}{ccc}
A^{(1,1)}_{u_1,v_1}&\cdots&A^{(p_1,1)}_{u_{p_1},v_1}\\
\vdots&&\vdots\\
A^{(1,p_2)}_{u_1,v_{p_2}}&\cdots&A^{(p_1,p_2)}_{u_{p_1},v_{p_2}}\\
\psi_1(x)&\cdots&\psi_{p_1}(x)
\end{array}
\right),\\
&Q_{\vec{u}',\vec{v}'}(x)=\sum_{i=1}^{p_1}B_i(x)\omega_{2,i}(x)\sim \det\left(\begin{array}{cccc}
A^{(1,1)}_{u_1',v_1'}&\cdots&A^{(p_1,1)}_{u_{p_1}',v_1'}&\varphi_1(x)\\
\vdots&&\vdots&\vdots\\
A^{(1,p_2)}_{u_1',v_{p_2}'}&\cdots&A^{(p_1,p_2)}_{u_{p_1}',v_{p_2}'}&\varphi_{p_2}(x)\\
\end{array}
\right),
\end{align*}
where 
\begin{align*}
\psi_i(x)=\omega_{1,i}(x)(1,x,\cdots,x^{u_{i}-1}),\quad \varphi_i(x)=\omega_{2,i}(x)(1,x,\cdots,x^{v_{i}'-1})',\quad A^{(a,b)}_{u_i,v_j}=\left(m_{l+k}^{(a,b)}\right)_{\substack{{k=0,\cdots,v_{j}-1}\\
		{l=0,\cdots,u_{i}-1}}}.
		\end{align*} Such formula implies that one can regard the block moment matrix as non-abelian moment matrix. Therefore, MOPs of type I, type II and mixed type are special non-abelian orthogonal polynomials discussed in \cite{af17,li21}.
Moreover, if the polynomials $ \{A_j\}_{j=1}^{p_1} $ and $ \{B_j\}_{j=1}^{p_2} $ are unique up to a multiplicative constant, then we call $ (\vec{u},\vec{v}) $ a normal pair of indices for the sets of weights $ \vec{\omega}_{1} $ and $ \vec{\omega}_{2} $. Therefore, it is always possible to choose a proper normalization to uniquely define MOPs of mixed type with regard to normal pair of indices. 
In agreement with formula \eqref{length}, we require that $|\vec{v}|=|\vec{u}|$, and $P_{\vec{u}+\vec{e}_a,\vec{v}}(x)$ and $Q_{\vec{u},\vec{v}+\vec{e}_b}(x)$ are desired formula satisfying orthonormal condition 
\begin{align*}
\int_{\mathbb{R}}P_{\vec{u}+\vec{e}_a,\vec{v}}(x)Q_{\vec{u},\vec{v}+\vec{e}_b}(x)dx=1.
\end{align*}
In above formula, 
\begin{align*}
\vec{e}_k=(0,\dots,1,\dots,0) \quad \text{where 1 is in the } k \text{th position}
\end{align*}
is the unit vector, and $1\leq a\leq p_1$ and $1\leq b\leq p_2$ are fixed integers.

If we further assume that $P_{\vec{u}+\vec{e}_a,\vec{v}}(x)$ and $Q_{\vec{u},\vec{v}+\vec{e}_b}$ have the same 
coefficients for the term $x^{u_a}\omega_{1,a}(x)$ and $x^{v_b}\omega_{2,b}(x)$, then by solving the linear system, we have
\begin{align*}
&P_{\vec{u},\vec{v}}(x)=\frac{(-1)^{\sum_{i=b+1}^{p_2}v_i}}{c_{\vec{u},\vec{v}}^{(a,b)}}\det\left(\begin{array}{ccccc}
A^{(1,1)}_{u_1,v_1}&\cdots&A_{u_a+1,v_1}^{(a,1)}&\cdots&A^{(p_1,1)}_{u_{p_1},v_1}\\
\vdots&&\vdots&&\vdots\\
A^{(1,p_2)}_{u_1,v_{p_2}}&\cdots&A^{(a,p_2)}_{u_a+1,v_{p_2}}&\cdots&A^{(p_1,p_2)}_{u_{p_1},v_{p_2}}\\
\psi_1(x)&\cdots&\tilde{\psi}_{a}(x)&\cdots&\psi_{p_1}(x)
\end{array}
\right),\\ &Q_{\vec{u},\vec{v}+\vec{e}_b}(y)=\frac{(-1)^{\sum_{j=a+1}^{p_1}u_j}}{c_{\vec{u},\vec{v}}^{(a,b)}}\det\left(\begin{array}{cccccc}
A^{(1,1)}_{u_1,v_1}&\cdots&A^{(p_1,1)}_{u_{p_1},v_1}&\varphi_1(x)\\
\vdots&&\vdots&\vdots\\
A^{(1,b)}_{u_1,v_b+1}&\cdots&A^{(p_1,b)}_{u_{p_1},v_b+1}&\tilde{\varphi}_b(x)\\
\vdots&&\vdots&\vdots\\
A^{(1,p_2)}_{u_1,v_{p_2}}&\cdots&A^{(p_1,p_2)}_{u_{p_1},v_{p_2}}&\varphi_{p_2}(x)\\
\end{array}
\right),
\end{align*}
where $A_{u_i,v_j}^{(a,b)}$ was defined before, 
\begin{align*}
	&\psi_i(x)=\omega_{1,i}(x)(1,x,\cdots,x^{u_{i}-1}),\,\,\,(i\ne a)&\tilde{\psi}_a(x)=\omega_{1,a}(x)(1,x,\cdots,x^{u_{a}}),
\\	
	&\varphi_j(x)=\omega_{2,j}(x)(1,x,\cdots,x^{v_{j}-1})',\,(j\ne b) & \tilde{\varphi}_b(x)=\omega_{2,i}(x)(1,x,\cdots,x^{v_{b}})',
\end{align*}
and
\begin{align*}
c_{\vec{u},\vec{v}}^{(a,b)}=\left(\det\left[\begin{array}{ccc}
A^{(1,1)}_{u_1,v_1}&\cdots&A^{(p_1,1)}_{u_{p_1},v_1}\\
\vdots&&\vdots\\
A^{(1,p_2)}_{u_1,v_{p_2}}&\cdots&A^{(p_1,p_2)}_{u_{p_1},v_{p_2}}
\end{array}
\right]
\det\left[\begin{array}{ccccc}
A^{(1,1)}_{u_1,v_1}&\dots&A^{(a,1)}_{u_a+1,v_1}&\dots&A^{(p_1,1)}_{u_{p_1},v_1}\\
\vdots&&\vdots&&\vdots\\
A_{u_1,v_{b}+1}^{(1,b)}&\dots&A^{(a,b)}_{u_a+1,v_b+1}&\dots&A_{u_{p_1},v_{b}+1}^{(p_1,b)}\\
\vdots&&\vdots&&\vdots\\
A^{(1,p_2)}_{u_1,v_{p_2}}&\dots&A^{(a,p_2)}_{u_a+1,v_{p_2}}&\dots&A^{(p_1,p_2)}_{u_{p_1},v_{p_2}}\end{array}
\right]
\right)^{1/2}.
\end{align*}

\subsection{A bi-orthogonal generalization of MOPs}\label{sec2.2}
This part is devoted to the bi-orthogonal generalization of MOPs with inner product \eqref{bf1}. Let's consider two pairs of different multi-indices $\vec{u}=(u_1,\cdots,u_{p_1})$, $\vec{v}=(v_1,\cdots,v_{p_2})$ and $\vec{u}'=(u_1',\cdots,u_{p_1}')$, $\vec{v}'=(v_1',\cdots,v_{p_2}')$, together with weights $\vec{\omega}_{1}=(\omega_{1,1},\cdots,\omega_{1,p_1})$ and $\vec{\omega}_{2}=(\omega_{2,1},\cdots,\omega_{2,p_2})$ supported on contours $\gamma_1$ and $\gamma_2$ respectively. Then one can introduce a coupling function $$\s(x,y):\gamma_1\times\gamma_2\to\mathbb{R}$$ such that for $1\leq a\leq p_1$ and $1\leq b\leq p_2$, bi-moments
\begin{align*}
m_{k,l}^{(a,b)}=\int_{\gamma_1\times \gamma_2}x^ky^l \s(x,y)\omega_{1,a}(x)\omega_{2,b}(y)dxdy
\end{align*}
exist and are finite. Therefore, we can define polynomials $ \{A_i\}_{i=1}^{p_1} $ together with its counterpart $\{B_j\}_{j=1}^{p_2} $ such that they satisfy the orthogonal relations
\begin{align*}
&\int_{\gamma_1\times\gamma_2}\left(
\sum_{i=1}^{p_1}A_i(x)\omega_{1,i}(x)
\right)\s(x,y)y^k\omega_{2,j}(y)dxdy=0,\quad k=0,\cdots,v_j-1,\quad j=1,\cdots,p_2,\\
&\int_{\gamma_1\times\gamma_2}x^k\omega_{1,j}(x)\s(x,y)\left(
\sum_{i=1}^{p_2}B_i(y)\omega_{2,i}(y)
\right)dxdy=0,\quad k=0,\cdots,u_j'-1,\quad j=1,\cdots,p_1.
\end{align*}
To uniquely determine these multiple bi-orthogonal polynomials (MBOPs), we follow our discussions about MOPs of mixed type, and the formal definition is given below. 
\begin{definition}
	Suppose we have two pairs of multi-indices $\vec{u}=(u_1,\cdots,u_{p_1})$, $\vec{v}=(v_1,\cdots,v_{p_2})$ and $\vec{u}'=(u_1',\cdots,u_{p_1}')$, $\vec{v}'=(v_1',\cdots,v_{p_2}')$ with $|\vec{u}|=|\vec{v}|$ and $ |\vec{u}'|=|\vec{v}'| $, together with two sets of weights $\vec{\omega_{1}}$ and $\vec{\omega_{2}}$, which are supported on contours $\gamma_1$ and $\gamma_2$ respectively. Fix integers $ 1\leq a\leq p_1 $ and $ 1\leq b\leq p_2 $ . If $\s(x,y)$ is a nice enough function from $\gamma_1\times \gamma_2$ to $\mathbb{R}$ so that all moments exist and are finite, then there are unique multiple bi-orthogonal functions 
	\begin{align*}
	P_{\vec{u}+\vec{e}_a,\vec{v}}(x)=\sum_{i=1}^{p_1}A_i(x)\omega_{1,i}(x), \text{ where deg $ A_i(x)\leq u_i-1 (i\ne a) $ and deg $ A_a(x)\leq u_a $},\\
	Q_{\vec{u}',\vec{v}'+\vec{e}_b}(y)=\sum_{i=1}^{p_2}B_i(y)\omega_{2,i}(y), \text{ where deg $ B_i(y)\leq v_i'-1 (i\ne b)$ and deg $ B_b(y)\leq v_b' $}
	\end{align*}
	satisfying multiple orthogonal relations  
	\begin{align}\label{mbops}
	\begin{aligned}
	\int_{\gamma_1\times \gamma_2}P_{\vec{u}+\vec{e}_a,\vec{v}}(x)\s(x,y)Q_{\vec{u}',\vec{v}'+\vec{e}_b}(y)dxdy=\left\{\begin{array}{ll}
	0 & \text{ if } \vec{u}+\vec{e}_a\leq \vec{u}',\\
	0 & \text{ if } \vec{v}\geq \vec{v}'+\vec{e}_b,\\
	1 & \text{ if } \vec{u}=\vec{u}' \text{ and } \vec{v}=\vec{v}'.
	\end{array}
	\right.
	\end{aligned}
	\end{align}
It is required that $P_{\vec{u}+\vec{e}_a,\vec{v}}(x)$ and $Q_{\vec{u},\vec{v}+\vec{e}_b}(y)$ have the same normalization factor. 
\end{definition}

By introducing
\begin{align*}
&\psi_{i}(x)=\omega_{1,i}(x)(1,x,\cdots,x^{u_i-1}), &\tilde{\psi}_{i}(x)=\omega_{1,i}(x)(1,x,\cdots,x^{u_i}), && i=1,\cdots,p_1, \\
&\varphi_i(x)=\omega_{2,i}(x)(1,x,\cdots,x^{v_i-1})', &\tilde{\varphi}_i(x)=\omega_{2,i}(x)(1,x,\cdots,x^{v_i})', && i=1,\cdots,p_2,
\end{align*}
and solving the orthogonal relations \eqref{mbops}, we know that 
\begin{align*}
&P_{\vec{u}+\vec{e}_a,\vec{v}}(x)=\frac{(-1)^{\sum_{i=b+1}^{p_2}v_i}}{c_{\vec{u},\vec{v}}^{(a,b)}}\det\left(\begin{array}{ccccc}
A^{(1,1)}_{u_1,v_1}&\dots&A^{(a,1)}_{u_a+1,v_1}&\dots&A^{(p_1,1)}_{u_{p_1},v_1}\\
\vdots&&\vdots&&\vdots\\
A^{(1,p_2)}_{u_1,v_{p_2}}&\cdots&A^{(a,p_2)}_{u_a+1,v_{p_2}}&\dots&A^{(p_1,p_2)}_{u_{p_1},v_{p_2}}\\
\psi_{1}(x)&\cdots&\tilde{\psi}_{a}(x)&\dots&\psi_{p_1}(x)\end{array}
\right),\\ &Q_{\vec{u},\vec{v}+\vec{e}_b}(y)=\frac{(-1)^{\sum_{j=a+1}^{p_1}u_j}}{c_{\vec{u},\vec{v}}^{(a,b)}}\det\left(\begin{array}{cccc}
A^{(1,1)}_{u_1,v_1}&\cdots&A^{(p_1,1)}_{u_{p_1},v_1}&\varphi_{1}(y)\\
\vdots&&\vdots&\vdots\\
A^{(1,b)}_{u_1,v_b+1}&\cdots&A^{(p_1,b)}_{u_{p_1},v_b+1}&\tilde{\varphi}_{b}(y)\\
\vdots&&\vdots&\vdots\\
A^{(1,p_2)}_{u_1,v_{p_2}}&\cdots&A^{(p_1,p_2)}_{u_{p_1},v_{p_2}}&\varphi_{p_2}(y)\end{array}
\right),
\end{align*}
where $ A_{u_i,v_j}^{(i,j)}=(m_{l,k}^{(i,j)})_{\substack{k=0,\dots,v_j-1\\ l=0,\dots,u_i-1}} $ and
\begin{align*}
c_{\vec{u},\vec{v}}^{(a,b)}=\left(\det\left[\begin{array}{ccc}
A^{(1,1)}_{u_1,v_1}&\cdots&A^{(p_1,1)}_{u_{p_1},v_1}\\
\vdots&&\vdots\\
A^{(1,p_2)}_{u_1,v_{p_2}}&\cdots&A^{(p_1,p_2)}_{u_{p_1},v_{p_2}}
\end{array}
\right]
\det\left[\begin{array}{ccccc}
A^{(1,1)}_{u_1,v_1}&\dots&A^{(a,1)}_{u_a+1,v_1}&\dots&A^{(p_1,1)}_{u_{p_1},v_1}\\
\vdots&&\vdots&&\vdots\\
A_{u_1,v_{b}+1}^{(1,b)}&\dots&A^{(a,b)}_{u_a+1,v_b+1}&\dots&A_{u_{p_1},v_{b}+1}^{(p_1,b)}\\
\vdots&&\vdots&&\vdots\\
A^{(1,p_2)}_{u_1,v_{p_2}}&\dots&A^{(a,p_2)}_{u_a+1,v_{p_2}}&\dots&A^{(p_1,p_2)}_{u_{p_1},v_{p_2}}\end{array}
\right]
\right)^{1/2}.
\end{align*}

\begin{remark}
	when $\vec{u}$ and $\vec{v}$ have only one index, multiple bi-orthogonal polynomials degenerate to normal bi-orthogonal polynomials, which have been well investigated. For example, the case $\s(x,y)=e^{-cxy}$ (where $c$ is a coupling constant) is related to a couple Hermitian matrix model and was studied in \cite{adler99,mehta02}. Moreover, the case $\s(x,y)=(x+y)^{-1}$ gives rise to the so-called Cauchy bi-orthogonal polynomials, which has attracted attention in different fields like random matrix, integrable systems and approximation theory \cite{lundmark05,bertola09,bertola10,li19}. 
\end{remark}

\subsection{Multiple symmetric bi-orthogonal polynomials and multiple skew-orthogonal polynomials}\label{sec2.3}
In this part, we prepare to give a definition of multiple skew-orthogonal polynomials. Due to the difficulties in skew-orthogonality, we firstly take a look at the multiple symmetric bi-orthogonal polynomials and then move to the skew-symmetric case. 
\subsubsection{Multiple symmetric bi-orthogonal polynomials}
In the symmetric case, we need to assume that multi-indices $\vec{u}$ and $\vec{v}$ as well as weights $\vec{\omega}_1$ and $\vec{\omega}_2$ are the same, that is, we have only one multiple index $\vec{v}=(v_1,\cdots,v_{p})$ and one family of weights $(\omega_{1},\cdots,\omega_{p})$ supported on $\gamma$. Moreover, the coupling function $\s(x,y):\gamma\times\gamma\to \mathbb{R}$ is a symmetric function, i.e. $\s(x,y)=\s(y,x)$. Therefore, moments under this setting could be written as 
\begin{align*}
m_{k,l}^{(i,j)}:=\int_{\gamma\times\gamma} x^ky^l\s(x,y)\omega_i(x)\omega_j(y)dxdy,
\end{align*}
and obviously $m_{k,l}^{(i,j)}=m_{l,k}^{(j,i)}$. Let $b\in\mathbb{Z}$ and $ 1\leq b\leq p $, we have a sequence of symmetric MBOPs $\{A_i(x)\}_{i=1}^p$, where deg $A_i\leq v_i-1$ $(i=1,\cdots,p, i\neq b)$ and deg $A_b\leq v_b$, such that corresponding linear form $ P_{\vec{v}}(x)=\sum_{i=1}^{p}A_i(x)\omega_i(x) $ satisfy the orthogonal relation
\begin{align}\label{smbops}
\int_{\gamma\times\gamma}P_{\vec{v}}(x)\s(x,y)P_{\vec{v}'}(y)dxdy=\left\{\begin{array}{ll}
0 & \text{ if $ \vec{v}+\vec{e}_b\leq\vec{v}' $ or $ \vec{v}\geq\vec{v}'+\vec{e}_b $},\\
1 & \text{ if $ \vec{v}=\vec{v}' $}.
\end{array}
\right.
\end{align}
In order to solve the relations, it is useful to write the following equivalent form
\begin{subequations}\label{sop}
	\begin{align}
	&\int_{\gamma\times\gamma}\left(
	\sum_{i=1}^{p}A_i(x)\omega_i(x)
	\right)\s(x,y)y^k\omega_j(y)dxdy=0,\quad k=0,\cdots,v_j-1,\quad j=1,\cdots,p\label{sop1}\\
	&\int_{\gamma\times\gamma}\left(
	\sum_{i=1}^{p}A_i(x)\omega_i(x)
	\right)\s(x,y)y^{v_b}\omega_b(y)dxdy=h_{\vec{v}}^{(b)}\ne0.\label{sop2}
	\end{align}
\end{subequations}
If we assume that $A_i(x)=a_{i,v_i-1}x^{v_i-1}+\cdots+a_{i,0}$ $(i\neq b, 1\leq i\leq p)$ and $A_b(x)=a_{b,v_b}x^{v_b}+\cdots+a_{b,0}$, then the above linear system is equivalent to
\begin{align}\label{coef}
\left(\begin{array}{ccccc}
A_{v_1,v_1}^{(1,1)}&\cdots&A_{v_b+1,v_1}^{(b,1)}&\cdots&A_{v_p,v_1}^{(p,1)}\\
\vdots&&\vdots&&\vdots\\
A_{v_1,v_b+1}^{(1,b)}&\cdots&A_{v_b+1,v_b+1}^{(b,b)}&\cdots&A_{v_p,v_b+1}^{(p,b)}\\
\vdots&&\vdots&&\vdots\\
A_{v_1,v_p}^{(1,p)}&\cdots&A_{v_b+1,v_p}^{(b,p)}&\cdots&A_{v_p,v_p}^{(p,p)}\end{array}
\right)\left(\begin{array}{c}
\alpha^{(1)}\\\vdots\\\alpha^{(b)}\\\vdots\\\alpha^{(p)}\end{array}
\right)=\left(\begin{array}{c}
0\\\vdots\\{\vec{e}_b}^\top\\\vdots\\0
\end{array}
\right),
\end{align}
where 
\begin{align*}
\alpha^{(i)}=(a_{i,0},\cdots,a_{i,v_i-1})',\,(i\neq b, 1\leq i\leq p)\quad \alpha^{(b)}=(a_{b,0},\cdots,a_{b,v_b})'
\end{align*} and ${\vec{e}_b}^\top$ is the transpose of $\vec{e}_b$. Therefore, we can obtain the following determinant form
\begin{align*}
P_{\vec{v}}(x)=\sum_{i=1}^p A_i(x)\omega_i(x)=\frac{(-1)^{\sum_{i=b+1}^{p}v_i}}{c_{\vec{v}}^{(b)}}\det\left(\begin{array}{ccccc}
A_{v_1,v_1}^{(1,1)}&\cdots&A_{v_b+1,v_1}^{(b,1)}&\cdots&A_{v_p,v_1}^{(p,1)}\\
\vdots&&\vdots&&\vdots\\
A_{v_1,v_p}^{(1,p)}&\cdots&A_{v_b+1,v_p}^{(b,p)}&\cdots&A_{v_p,v_p}^{(p,p)}\\
\psi_1(x)&\cdots&\psi_b(x)&\cdots&\psi_p(x)
\end{array}
\right),
\end{align*}
where $\psi_i(x)=\omega_i(x)(1,\cdots,x^{v_i-1})$ $(i\neq b, 1\leq i\leq p)$ and $\psi_b(x)=\omega_b(x)(1,\cdots,x^{v_b})$. If we denote 
\begin{align*}
\tau_{(v_1,\cdots,v_p)}=\det\left(\begin{array}{ccc}
A_{v_1,v_1}^{(1,1)}&\cdots&A_{v_p,v_1}^{(p,1)}\\
\vdots&&\vdots\\
A_{v_1,v_p}^{(1,p)}&\cdots&A_{v_p,v_p}^{(p,p)}
\end{array}
\right),\quad A_{\alpha,\beta}^{(i,j)}=A_{\beta,\alpha}^{(j,i)}.
\end{align*}
then we have $c_{\vec{v}}^{(b)}=(\tau_{(v_1,\cdots,v_p)}\tau_{(v_1,\cdots,v_b+1,\cdots,v_p)})^{1/2}$ and $h_{\vec{v}}^{(b)}$ in \eqref{sop2} could be expressed by $(\tau_{(v_1,\cdots,v_b+1,\cdots,v_p)}/\tau_{(v_1,\cdots,v_p)})^{1/2}$.  
\subsubsection{Multiple skew-orthogonal polynomials}
Let's consider a skew-symmetric kernel $\s(x,y)=-\s(y,x)$ such that 
\begin{align*}
m_{k,l}^{(a,b)}:=\int_{\gamma\times\gamma} x^ky^l\s(x,y)\omega_a(x)\omega_b(y)dxdy=-\int_{\gamma\times\gamma}x^ly^k\s(x,y)\omega_b(x)\omega_a(y)dxdy=-m_{l,k}^{(b,a)}.
\end{align*}
Then for a multi-index $\vec{v}=(v_1,\cdots,v_p)$ and a sequence of weights $(\omega_1,\cdots,\omega_p)$, we can define corresponding polynomials $(R_1(x),\cdots,R_p(x))$ where deg $R_i\leq v_i-1$ ($i=1,\cdots,p$). Since our primary consideration is to seek for the linear form $\sum_{i=1}^p R_i(x)\omega_i(x)$, which is simultaneously skew orthogonal with respect to several weights, we first consider the multiple skew orthogonal relations 
\begin{align}\label{eq2.14}
\int_{\gamma\times\gamma} \left(
\sum_{i=1}^p R_i(x)\omega_i(x)
\right)\s(x,y)y^k\omega_j(y)dxdy=0,\quad k=0,\cdots,v_j-1,\quad j=1,\cdots,p.
\end{align}

If we denote $R_i(x)=a_{i,v_i-1}x^{v_i-1}+\cdots+a_{i,0}$, then equation \eqref{eq2.14} implies
\begin{align*}
\left(\begin{array}{ccc}
A_{v_1,v_1}^{(1,1)}&\cdots&A_{v_p,v_1}^{(p,1)}\\
\vdots&&\vdots\\
A_{v_1,v_p}^{(1,p)}&\cdots&A_{v_p,v_p}^{(p,p)}
\end{array}
\right)\left(\begin{array}{c}
\alpha^{(1)}\\\vdots\\\alpha^{(p)}
\end{array}
\right)=0,
\end{align*}
where $ A_{u_i,v_j}^{(i,j)}=(m_{l,k}^{(i,j)})_{\substack{k=0,\dots,v_j-1\\ l=0,\dots,u_i-1}} $ and $ \alpha^{(i)}=(a_{i,0},\cdots,a_{i,v_i-1})' $.
Since the matrix is skew symmetric, one knows that non-trivial solutions for $\alpha^{(i)}$ always exist when $v_1+\cdots+v_p$ is odd (i.e. $|\vec{v}|$ is odd). Therefore, it is a key observation that MSOPs are valid only for $|\vec{v}|$ being odd. The normalization condition is then given by
\begin{align*}
\int_{\gamma\times\gamma}\left(\sum_{i=1}^p R_i(x)\omega_i(x)
\right)\s(x,y)y^{v_b}\omega_b(y)dxdy=h_{\vec{v}}^{(b)}\ne0,
\end{align*}
where $ b $ is a fixed integer between $1$ and $p$.
To conclude, we have the following definition for multiple skew-orthogonal polynomials.
\begin{definition}\label{msopdef}
	Given a multi-index $\vec{v}=(v_1,\cdots,v_p)$ such that $|\vec{v}|=v_1+\cdots+v_p$ is odd. If there are $p$ different weights $(\omega_1,\cdots,\omega_p)$ supported on $\gamma$ and $\s(x,y)$ is a skew-symmetric function from $\gamma\times\gamma$ to $\mathbb{R}$ so that all moments are finite, then for a fixed integer $b\in[1,p]$, there exist multiple skew orthogonal polynomials $R_1(x),\cdots,R_p(x)$ and $\tilde{R}_b(x)$, such that
	\begin{align}\label{msop}
	\begin{aligned}
	&\int_{\gamma\times\gamma}\left(
	\sum_{i=1}^{p}R_i(x)\omega_i(x)
	\right)\s(x,y)y^j\omega_k(y)dxdy=0,\quad j=0,\cdots,v_k-1,\quad k=1,\cdots,p,\\
	&\int_{\gamma\times\gamma}\left(
	\sum_{i=1}^{p}R_i(x)\omega_i(x)
	\right)\s(x,y)\left(
	\sum_{\substack{i=1\\i\ne b}}^{p}R_i(y)\omega_i(y)+\tilde{R}_b(y)\omega_b(y)
	\right)dxdy=1,
	\end{aligned}
	\end{align}
	where deg $R_i(x)\leq v_i-1$ $(i=1,\cdots,p)$, and deg $\tilde{R}_b(x)\leq v_b$. Here we assume that coefficients in the highest order terms of $R_b$ and $\tilde{R}_b$ are the same.
\end{definition}
\begin{remark}\label{rem1}
	Different from the orthogonal relations for symmetric MBOPs \eqref{smbops}, the skew inner product of MSOPs and itself is equal to zero, i.e. 
\begin{align*}
\int_{\gamma\times \gamma} \left(\sum_{i=1}^{p}R_i(x)\omega_i(x)\right)\s(x,y)\left(\sum_{i=1}^{p}R_i(y)\omega_i(y)
\right)dxdy=0.
\end{align*}	
Therefore, the skew-orthogonality is not affected by the scaling  $\tilde{R}_b(y)\to\tilde{R}_b(y)+\alpha R_b(y)$ for all $\alpha\in\mathbb{R}$. Therefore, for later convenience, we denote $(R_1(x),\cdots,R_p(x),\tilde{R}_b(x))$ as a family of multiple skew-orthogonal polynomials, and set the coefficient of $x^{v_b-1}\omega_b(x)$ in $\sum_{i=1}^pR_i(x)\omega_i(x)+\tilde{R}_b(x)\omega_b(x)$ as $0$. 
\end{remark}
By assuming that coefficients in the highest order terms of $R_b(x)$ and $\tilde{R}_b(x)$ are the same, equation \eqref{msop} has a unique solution and we have
\begin{subequations}
	\begin{align}\label{det1}
	R_{\vec{v}}^{(b)}(x)&:=\sum_{i=1}^pR_i(x)\omega_i(x)={\frac{(-1)^{\sum_{i=b+1}^{p}v_i}}{c_{\vec{v}}^{(b)}}}
	\det\left(\begin{array}{ccc}
	A_{v_1,v_1}^{(1,1)}&\cdots&A_{v_p,v_1}^{(p,1)}\\
	\vdots&&\vdots\\
	A_{v_1,v_b-1}^{(1,b)}&\cdots&A_{v_p,v_b-1}^{(p,b)}\\
	\vdots&&\vdots\\
	A_{v_1,v_p}^{(1,p)}&\cdots&A_{v_p,v_p}^{(p,p)}\\
	\psi_1(x)&\cdots&\psi_p(x)
	\end{array}
	\right),\\\label{det2}
	\tilde{R}_{\vec{v}}^{(b)}(x)&:=\sum_{i=1}^{p}R_i(y)\omega_i(y)+\tilde{R}_b(y)\omega_b(y)\\
	&={\frac{1}{c_{\vec{v}}^{(b)}}}
	\det\left(\begin{array}{cccccc}
	A_{v_1,v_1}^{(1,1)}&\cdots&A_{v_b-1,v_1}^{(b,1)}&\cdots&A_{v_p,v_1}^{(p,1)}&\psi_1(y)'\\
	\vdots&&\vdots&&\vdots&\vdots\\
	A_{v_1,v_b-1}^{(1,b)}&\cdots&A_{v_b-1,v_b-1}^{(b,b)}&\cdots&A_{v_p,v_b-1}^{(p,b)}&\tilde{\psi}_b(y)'\\
	\vdots&&\vdots&&\vdots&\vdots\\
	A_{v_1,v_p}^{(1,p)}&\cdots&A_{v_b-1,v_p}^{(b,p)}&\cdots&A_{v_p,v_p}^{(p,p)}&\psi_p(y)'\\
	M_{v_1,v_b}^{(1,b)}&\cdots&M_{v_b-1,v_b}^{(b,b)}&\cdots&M_{v_p,v_b}^{(p,b)}&y^{v_b}\omega_b(y)
	\end{array}
	\right),
	\end{align}
\end{subequations}
where
\begin{align*}
\psi_i(x)=\omega_i(x)(1,\cdots,x^{v_i-1}),\,(i=1,\cdots,p),\quad\tilde{\psi}_b(x)=\omega_b(x)(1,\cdots,x^{v_b-2}), \quad M_{v_i,v_j}^{(k,l)}=(m_{i,v_j}^{(k,l)})_{i=0}^{v_i-1}
\end{align*}
and the normalization factor $ c^{(b)}_{\vec{v}} $ is given by
\begin{align}\label{det3}
(c^{(b)}_{\vec{v}})^{2}=\det\left(\begin{array}{ccc}
A_{v_1,v_1}^{(1,1)}&\cdots&A_{v_p,v_1}^{(p,1)}\\
\vdots&&\vdots\\
A_{v_1,v_b-1}^{(1,b)}&\cdots&A_{v_p,v_b-1}^{(p,b)}\\
M_{v_1,v_b}^{(1,b)}&\cdots&M_{v_p,v_b}^{(p,b)}\\
A_{v_1,v_{b+1}}^{(1,b+1)}&\cdots&A_{v_p,v_{b+1}}^{(p,b+1)}\\
\vdots&&\vdots\\
A_{v_1,v_p}^{(1,p)}&\cdots&A_{v_p,v_p}^{(p,p)}\end{array}\right)
\det\left(\begin{array}{ccccc}
A_{v_1,v_1}^{(1,1)}&\cdots&A_{v_b-1,v_1}^{(b,1)}&\cdots&A_{v_p,v_1}^{(p,1)}\\
\vdots&&\vdots&&\vdots\\
A_{v_1,v_b-1}^{(1,b)}&\cdots&A_{v_b-1,v_b-1}^{(b,b)}&\cdots&A_{v_p,v_b-1}^{(p,b)}\\
\vdots&&\vdots&&\vdots\\
A_{v_1,v_p}^{(1,p)}&\cdots&A_{v_b-1,v_p}^{(b,p)}&\cdots&A_{v_p,v_p}^{(p,p)}
\end{array}
\right).
\end{align}

\section{Pfaffian form of multiple skew-orthogonal polynomials}\label{sec3}

As is known, skew-orthogonal polynomials have Pfaffian expressions which are widely used in integrable systems in terms of Pfaffian tau-functions \cite{adler992,adler02,kodama10}.
In this section we plan to express MSOPs by Pfaffian, and investigate its evolution when time parameters are introduced. The 2-component case has to be a primary consideration since the multiple-component case can be easily generalized from the 2-component case. 
Therefore, from now on, we assume the index set $v=(v_1,v_2)$ and $v_1+v_2$ is odd. 

\subsection{Pfaffian expressions for MSOPs}
First, we note that determinant expressions in \eqref{det1}, \eqref{det2} and \eqref{det3} can be alternatively written in terms of Pfaffians, which is stated as follows.
\begin{proposition}\label{prop3.1}
For multiple skew orthogonal polynomials $(R_1(x),R_2(x),\tilde{R}_2(x))$, we have
\begin{subequations}
\begin{align}\label{pf1}
&R_{(v_1,v_2)}^{(2)}(x):=R_1(x)\omega_1(x)+R_2(x)\omega_2(x)={\frac{1}{d^{(2)}_{\vec{v}}}}\Pf\left(\begin{array}{ccc}
A_{v_1,v_1}^{(1,1)}&A_{v_2,v_1}^{(2,1)}&-\psi_1(x)\\
A_{v_1,v_2}^{(1,2)}&A_{v_2,v_2}^{(2,2)}&-\psi_2(x)\\
\psi_1(x)&\psi_2(x)&0\end{array}
\right),\\
&\tilde{R}_{(v_1,v_2)}^{(2)}(x)=R_1(x)\omega_1(x)+R_2(x)\omega_2(x)+\tilde{R}_2(x)\omega_2(x)\nonumber\\
&\qquad\qquad={\frac{1}{d^{(2)}_{\vec{v}}}}\Pf\left(\begin{array}{cccc}
A_{v_1,v_1}^{(1,1)}&A_{v_2-1,v_1}^{(2,1)}&(M_{v_2,v_1}^{(2,1)})^\top&-\psi_1(x)\\
A_{v_1,v_2-1}^{(1,2)}&A_{v_2-1,v_2-1}^{(2,2)}&(M_{v_2,v_2-1}^{(2,2)})^\top&-\tilde{\psi}_2(x)\\
M_{v_1,v_2}^{(1,2)}&M_{v_2-1,v_2}^{(2,2)}&0&-x^{v_2}\omega_2(x)\\
\psi_1(x)&\tilde{\psi}_2(x)&x^{v_2}\omega_2(x)&0\label{pf2}
\end{array}
\right),
\end{align}
\end{subequations}
with the normalization factor
\begin{align*}
d^{(2)}_{\vec{v}}=\left(
\Pf\left(\begin{array}{cc}
A_{v_1,v_1}^{(1,1)}&A_{v_2-1,v_1}^{(2,1)}\\
A_{v_1,v_2-1}^{(1,2)}&A_{v_2-1,v_2-1}^{(2,2)}
\end{array}
\right)
\Pf\left(\begin{array}{cc}
A_{v_1,v_1}^{(1,1)}&A_{v_2+1,v_1}^{(2,1)}\\
A_{v_1,v_2+1}^{(1,2)}&A_{v_2+1,v_2+1}^{(2,2)}
\end{array}
\right)
\right)^{1/2}.
\end{align*}
\end{proposition}
\begin{proof}
Here we give a clear explanation for formula \eqref{det1}, while \eqref{det2} can be similarly verified. By applying Jacobi determinant identity \footnote{The Jacobi identity, also referred to as the Desnanot-Jacobi identity, is applied for an arbitrary matrix $M=(m_{i,j})_{i,j=1}^N$
\begin{align*}
|M|\times |M^{a,b}_{c,d}|=|M^{a}_c|\times |M_d^b|-|M_a^d|\times |M_c^b|,
\end{align*}
where $|M_{i_1,\cdots,i_r}^{j_1,\cdots,j_r}|$ stands for the determinant of the the matrix obtained from $M$ by deleting its $(i_1,\cdots,i_r)$-th rows and $(j_1,\cdots,j_r)$-th columns.} to 
\begin{align*}
\left(
\begin{array}{ccccccc}
m_{0,0}^{(1,1)}&\cdots&m_{v_1-1,0}^{(1,1)}&m_{0,0}^{(2,1)}&\cdots&m_{v_2-1,0}^{(2,1)}&-\omega_1(x)\\
\vdots&&\vdots&\vdots&&\vdots&\vdots\\
m_{0,v_1-1}^{(1,1)}&\cdots&m_{v_1-1,v_1-1}^{(1,1)}&m_{0,v_2-1}^{(1,2)}&\cdots&m_{v_2-1,v_1-1}^{(2,1)}&-x^{v_1-1}\omega_1(x)\\
m_{0,0}^{(1,2)}&\cdots&m_{v_1-1,0}^{(1,2)}&m_{0,0}^{(2,2)}&\cdots&m_{v_2-1,0}^{(2,2)}&-\omega_2(x)\\
\vdots&&\vdots&\vdots&&\vdots&\vdots\\
m_{0,v_2-1}^{(1,2)}&\cdots&m_{v_1-1,v_2-1}^{(1,2)}&m_{0,v_2-1}^{(2,2)}&\cdots&m_{v_2-1,v_2-1}^{(2,2)}&-x^{v_2-1}\omega_2(x)\\
\omega_1(x)&\cdots&x^{v_1-1}\omega_1(x)&\omega_2(x)&\cdots&x^{v_2-1}\omega_2(x)&0
\end{array}\right)
\end{align*}
for last two rows and columns, and noting that the determinant of  an odd-order skew-symmetric matrix is zero, we obtain
\begin{align*}
R_{(v_1,v_2)}^{(2)}(x)={\frac{1}{c^{(2)}_{(v_1,v_2)}}}
\Pf\left(\begin{array}{cc}
A_{v_1,v_1}^{(1,1)}&A_{v_2-1,v_1}^{(2,1)}\\
A_{v_1,v_2-1}^{(1,2)}&A_{v_2-1,v_2-1}^{(2,2)}\end{array}
\right)
\Pf\left(\begin{array}{ccc}
A_{v_1,v_1}^{(1,1)}&A_{v_2,v_1}^{(2,1)}&-\psi_1(x)\\
A_{v_1,v_2}^{(1,2)}&A_{v_2,v_2}^{(2,2)}&-\psi_2(x)\\
\psi_1(x)&\psi_2(x)&0\end{array}
\right)
\end{align*}
Moreover, by applying determinant identity to $c^{(2)}_{(v_1,v_2)}$ in \eqref{det3} for the first determinant, we have
\begin{align*}
c^{(2)}_{(v_1,v_2)}=\Pf\left(\begin{array}{cc}
A_{v_1,v_1}^{(1,1)}&A_{v_2-1,v_1}^{(2,1)}\\
A_{v_1,v_2-1}^{(1,2)}&A_{v_2-1,v_2-1}^{(2,2)}\end{array}
\right)^{3/2}
\Pf\left(\begin{array}{cc}
A_{v_1,v_1}^{(1,1)}&A_{v_2+1,v_1}^{(2,1)}\\
A_{v_1,v_2+1}^{(1,2)}&A_{v_2+1,v_2+1}^{(2,2)}\end{array}
\right)
^{1/2},
\end{align*}
and thus the proof is complete.
\end{proof}

Therefore, one can use Hirota's Pfaffian notations \cite{hirota04} to make these expressions more compact. 
If we denote 
\begin{align*}
\pf(i^{(k)}, j^{(l)})=m_{i,j}^{(k,l)}, \quad
\pf(i^{(k)},x)=\omega_k(x)x^i,\quad
 (k,l=1,2),
\end{align*}
then equations \eqref{pf1} and \eqref{pf2} could be equivalently expressed by
\begin{align}\label{msop1}
\begin{aligned}
&R^{(2)}_{(v_1,v_2)}(x)={\frac{1}{d^{(2)}_{(v_1,v_2)}}}\Pf(0^{(1)},\cdots,v_1-1^{(1)},0^{(2)},\cdots,v_2-1^{(2)},x),\\
&\tilde{R}^{(2)}_{(v_1,v_2)}(x)={\frac{1}{d^{(2)}_{(v_1,v_2)}}}\Pf(0^{(1)},\cdots,v_1-1^{(1)},0^{(2)},\cdots,v_2-2^{(2)},v_2^{(2)},x),
\end{aligned}
\end{align}
where $d^{(2)}_{(v_1,v_2)}=(\tau_{(v_1,v_2-1)}\tau_{(v_1,v_2+1)})^{1/2}$ and $$\tau_{(v_1,v_2-1)}=\Pf(0^{(1)},\cdots,v_1-1^{(1)},0^{(2)},\cdots,v_2-2^{(2)}).$$
According to our discussions in last section, there should be another family of multiple skew orthogonal polynomials $(R_{(v_1,v_2)}^{(1)}(x),\tilde{R}_{(v_1,v_2)}^{(1)}(x))$ such that
\begin{align}\label{msop2}
\begin{aligned}
&R_{(v_1,v_2)}^{(1)}(x)=\frac{1}{d_{(v_1,v_2)}^{(1)}}\Pf(0^{(1)},\cdots,v_1-1^{(1)},0^{(2)},\cdots,v_2-1^{(2)},x),\\
&\tilde{R}_{(v_1,v_2)}^{(1)}(x)=\frac{1}{d_{(v_1,v_2)}^{(1)}}\Pf(0^{(1)},\cdots,v_1-2^{(1)},v_1^{(1)},0^{(2)},\cdots,v_2-1^{(2)},x),
\end{aligned}
\end{align}
where $d_{(v_1,v_2)}^{(1)}=\left(
\tau_{(v_1-1,v_2)}\tau_{(v_1+1,v_2)}
\right)^{1/2}$. Moreover, from \eqref{msop1} and \eqref{msop2}, one knows that $R_{(v_1,v_2)}^{(1)}(x)$ and $R_{(v_1,v_2)}^{(2)}(x)$ are the same up to a normalization factor. 

By using Pfaffian notations, the skew orthogonal relations given by Definition \ref{msopdef} have the following equivalent descriptions.

\begin{proposition}\label{prop3.2}
$R_{(v_1,v_2)}^{(1)}(x)$  and $ R_{(v_1,v_2)}^{(2)}(x)$ are simultaneously skew orthogonal with $\tilde{R}_{(v_1,v_2)}^{(1)}(x)$ and $\tilde{R}_{(v_1,v_2)}^{(2)}(x)$, i.e. 
\begin{subequations}
\begin{align}
&\langle R_{(v_1,v_2)}^{(1)}(x),R_{(u_1,u_2)}^{(1)}(y)\rangle=0,\label{3.5a}\\
& \langle R_{(v_1,v_2)}^{(1)}(x),\tilde{R}_{(u_1,u_2)}^{(1)}(y)\rangle=\left\{\begin{array}{ll}
0,&\text{ if $u_1<v_1$ and $u_2\leq v_2$},\\
1,&\text{ if $u_1=v_1$ and $u_2=v_2$},\\
\end{array}\right.\label{3.5b}\\
& \langle R_{(v_1,v_2)}^{(1)}(x),\tilde{R}_{(u_1,u_2)}^{(2)}(y)\rangle=\left\{\begin{array}{ll}
0,&\text{ if $u_1\leq v_1$ and $u_2< v_2$},\\
d_{(v_1,v_2)}^{(2)}/d_{(v_1,v_2)}^{(1)},&\text{ if $u_1=v_1$ and $u_2=v_2$}.\label{3.5c}\\
\end{array}\right.
\end{align}
\end{subequations}
\end{proposition}
\begin{proof}
Since equations \eqref{3.5a} and \eqref{3.5b} have been shown in the last section, we prove the third  equation \eqref{3.5c} by using Pfaffian notations. Taking the Pfaffian expressions \eqref{msop1} and \eqref{msop2} into the skew inner product, we have
\begin{align}
\begin{aligned}\label{sum}
{d_{(v_1,v_2)}^{(1)}d_{(u_1,u_2)}^{(2)}}&\langle R_{(v_1,v_2)}^{(1)}(x),\tilde{R}_{(u_1,u_2)}^{(2)}(y)\rangle
\\
&=\sum_{i\in I_1}\sum_{j\in I_2}(-1)^{|i|+|j|}\Pf(I_1\backslash\{i\})\Pf(I_2\backslash\{j\})\langle \pf(i,x),\pf(j,y)\rangle,
\end{aligned}
\end{align}
where $I_1=\{0^{(1)},\cdots,v_1-1^{(1)},0^{(2)},\cdots,v_2-1^{(2)}\}$, $I_2=\{0^{(1)},\cdots,u_1-1^{(1)},0^{(2)},\cdots,u_2-2^{(2)},u_2^{(2)}\}$, and $|i|(|j|)$ represents the position of $i(j)$ in the set $I_1(I_2)$. By noting that 
\begin{align*}
\langle \pf(i^{(k)},x),\pf(j^{(l)},y)\rangle=\int_{\gamma\times\gamma}x^i\s(x,y)y^j\omega_k(x)\omega_l(y)dxdy=\pf(i^{(k)},j^{(l)}),
\end{align*}
 then the right hand side in \eqref{sum} is equal to
\begin{align}\label{3.6}
\sum_{j\in I_2}(-1)^{|j|}\Pf(I_1,j)\Pf(I_2\backslash\{j\}).
\end{align}
It is known that a Pfaffian is equal to zero if two indices in a Pfaffian are equal. Therefore, if $u_1\leq v_1$ and $u_2<v_2$, we know that $j\in I_1$ and above formula is identically zero. If $u_1=v_1$ and $u_2=v_2$, then only when $j=u_2^{(2)}=v_2^{(2)}$, the term is non-zero. In such a case, equation  \eqref{3.6} is equal to 
$
\tau_{(v_1,v_2+1)}\tau_{(v_1,v_2-1)}.
$
Therefore, we have
\begin{align*}
\langle R_{(v_1,v_2)}^{(1)}(x),\tilde{R}_{(v_1,v_2)}^{(2)}(y)\rangle=\frac{\tau_{(v_1,v_2+1)}\tau_{(v_1,v_2-1)}}{d_{(v_1,v_2)}^{(1)}d_{(v_1,v_2)}^{(2)}}=\frac{d_{(v_1,v_2)}^{(2)}}{d_{(v_1,v_2)}^{(1)}}.
\end{align*}
\end{proof}

\subsection{Semi-classical weights and deformed MSOPs}
Let's consider semi-classical weight functions. By introducing parameters $\mt:=(t_1,t_2,\cdots)$  and $\ms:=(s_1,s_2,\cdots)$ into weights $\omega_1$ and $\omega_2$ respectively such that
\begin{align*}
\omega_1(x;\mt)=\omega_1(x)\exp\left(
\sum_{i=1}^\infty t_ix^i
\right),\quad \omega_2(x;\ms)=\omega_2(x)\exp\left(
\sum_{i=1}^\infty s_ix^i
\right),
\end{align*}
we have
\begin{align*}
\p_{t_i}\omega_1(x;\mt)=x^i\omega_1(x;\mt),\quad \p_{s_i}\omega_2(x;\ms)=x^i\omega_2(x;\ms),\quad \p_{t_i}\omega_2(x;\ms)=\p_{s_i}\omega_1(x;\mt)=0.
\end{align*}
Moreover, moments are time-dependent and they obey the following deformations.
\begin{proposition}\label{prop3.3}
For moments $\{m_{a,b}^{(k,l)},\,k,l=1,2\}$, they have following evolutions
\begin{align*}
&\p_{t_i}m_{a,b}^{(1,1)}=m_{a+i,b}^{(1,1)}+m_{a,b+i}^{(1,1)},&&\p_{t_i}m_{a,b}^{(1,2)}=m_{a+i,b}^{(1,2)},&\p_{t_i}m_{a,b}^{(2,2)}=0,\\
&\p_{s_i}m_{a,b}^{(2,2)}=m_{a+i,b}^{(2,2)}+m_{a,b+i}^{(2,2)},&&\p_{s_i}m_{a,b}^{(1,2)}=m_{a,b+i}^{(1,2)},&\p_{s_i}m_{a,b}^{(1,1)}=0.
\end{align*}
Equivalently, in Pfaffian notations we have
\begin{align*}
&\p_{t_i}\pf(a^{(1)},b^{(1)})=\pf(a+i^{(1)},b^{(1)})+\pf(a^{(1)},b+i^{(1)}),&&\p_{t_i}\pf(a^{(1)},b^{(2)})=\pf(a+i^{(1)},b^{(2)}),\\
&\p_{s_i}\pf(a^{(2)},b^{(2)})=\pf(a+i^{(2)},b^{(2)})+\pf(a^{(2)},b+i^{(2)}),&&\p_{s_i}\pf(a^{(1)},b^{(2)})=\pf(a^{(1)},b+i^{(2)}),\\
&\p_{t_i}\pf(a^{(2)},b^{(2)})=\p_{s_i}\pf(a^{(1)},b^{(1)})=0.
\end{align*}
\end{proposition}

\begin{proof}
Let's prove $\p_{t_i}\pf(a^{(1)},b^{(1)})=\pf(a+i^{(1)},b^{(1)})+\pf(a^{(1)},b+i^{(1)})$, and other cases could be similarly verified. We first have
\begin{align*}
\p_{t_i}\pf(a^{(1)},b^{(1)})=\p_{t_i}\int_{\gamma\times\gamma} x^a\s(x,y)y^b\omega_1(x;\mt)\omega_1(y;\mt)dxdy.
\end{align*}
By noting that the moment is finite and weight $\omega_1(x;\mt)$ is smooth with respect to $\mt$, we know that the order of derivative and integration could be exchanged. Therefore, the above formula is equal to
\begin{align*}
\int_{\gamma\times\gamma} x^a\s(x,y)y^b (x^i+y^i)\omega_1(x;\mt)\omega_1(x;\mt)dxdy,
\end{align*}
which is exactly $\pf(a+i^{(1)},b^{(1)})+\pf(a^{(1)},b+i^{(1)})$.
\end{proof}

With such time parameters introduced, we can use derivative formulas for Wronskian type Pfaffians to deduce deformation relations for the linear forms of MSOPs.
\begin{proposition}\label{td}
 $R_{(v_1,v_2)}^{(i)}(x;\mt,\ms)$ and $\tilde{R}_{(v_1,v_2)}^{(i)}(x;\mt,\ms)$ $(i=1,2)$ have the following derivative relations
\begin{align}\label{te}
\begin{aligned}
\p_{t_1}\left(d_{(v_1,v_2)}^{(1)}R_{(v_1,v_2)}^{(1)}(x;\mt,\ms)\right)=d_{(v_1,v_2)}^{(1)}\tilde{R}_{(v_1,v_2)}^{(1)}(x;\mt,\ms),\\
\p_{s_1}\left(d_{(v_1,v_2)}^{(2)}R_{(v_1,v_2)}^{(2)}(x;\mt,\ms)\right)=d_{(v_1,v_2)}^{(2)}\tilde{R}_{(v_1,v_2)}^{(2)}(x;\mt,\ms).
\end{aligned}
\end{align}
\begin{proof}
Since $\mt$ and $\ms$ are dual to each other, we only prove the $t_1$-derivative formula. By using Pfaffian notations, it is equivalent to show that
\begin{align}\label{derivative1}
\begin{aligned}
\p_{t_1}&\pf(0^{(1)},\cdots,v_1-1^{(1)},0^{(2)},\cdots,v_2-1^{(2)},x)\\
&=\pf(0^{(1)},\cdots,v_1-2^{(1)},v_1^{(1)},0^{(2)},\cdots,v_2-1^{(2)},x).
\end{aligned}
\end{align}
If we introduce the index sets 
\begin{align*}\text{$I_1=\{0^{(1)},\cdots,v_1-1^{(1)}\}$, $\tilde{I}_1=\{0^{(1)},\cdots,v_1-2^{(1)},v_1^{(1)}\}$,  $I_2=\{0^{(2)},\cdots,v_2-1^{(2)}\}$},
\end{align*}
then  by expanding the Pfaffian, the left hand side in \eqref{derivative1} is equal to 
\begin{align*}
\p_{t_1}\left(
\sum_{i\in I_1}(-1)^{|i|}\pf(I_1\backslash\{i\},I_2)\pf(i,x)+\sum_{j\in I_2}(-1)^{|j|}\pf(I_1,I_2\backslash\{j\})\pf(j,x)
\right).
\end{align*}
By using derivative formula for Wronskian-type Pfaffians (see Appendix for details), we know that the first term is equal to 
\begin{align}\label{de2}
\begin{aligned}
&\sum_{i\in I_1\backslash\{0^{(1)}\}}(-1)^{|i|}\pf(I_1\backslash\{i-1\},I_2)\pf(i,x)+\sum_{i\in \tilde{I}_1\backslash\{v_1^{(1)}\}}(-1)^{|i|}\pf(\tilde{I}_1\backslash\{i\},I_2)\pf(i,x)\\
&+\sum_{i\in I_1}(-1)^{|i|}\pf(I_1\backslash\{i\},I_2)\pf(i+1,x),
\end{aligned}
\end{align}
and the second equals
\begin{align*}
\sum_{j\in I_2}(-1)^{|j|}\pf(\tilde{I}_1,I_2\backslash\{j\})\pf(j,x).
\end{align*}
A cancellation can be applied to the first and third term in \eqref{de2}. Thus, by combining these equations, we obtain 
\begin{align*}
	\sum_{i\in \tilde{I}_1}(-1)^{|i|}\pf(\tilde{I}_1\backslash\{i\},I_2)\pf(i,x)+\sum_{j\in I_2}(-1)^{|j|}\pf(\tilde{I}_1,I_2\backslash\{j\})\pf(j,x),
\end{align*}
which is exactly the expansion of the right hand side in \eqref{derivative1}.
\end{proof}

\end{proposition}

In despite of time evolutions for the linear forms of MSOPs, there should be spectral problems between $R_{(v_1,v_2)}^{(i)}(x)$ and $\tilde{R}_{(v_1,v_2)}^{(i)}(x)$ $(i=1,2)$, which are prominent in the derivations of integrable hierarchies. In below, we use Pfaffian identities to characterize spectral problems. 

\begin{proposition}\label{prop3.5}
$R_{(v_1,v_2)}^{(i)}(x;\mt,\ms)$ and $\tilde{R}_{(v_1,v_2)}^{(i)}(x;\mt,\ms)$ $(i=1,2)$ satisfy the following recurrence relations
\begin{subequations}
\begin{align}
\tau_{(v_1,v_2-1)}d_{(v_1+1,v_2+1)}^{(2)}&R_{(v_1+1,v_2+1)}^{(2)}(x)=\tau_{(v_1+1,v_2)}d_{(v_1,v_2)}^{(2)}\tilde{R}^{(2)}_{(v_1,v_2)}(x)\label{sub1}\\
&-\p_{s_1}\tau_{(v_1+1,v_2)}d_{(v_1,v_2)}^{(2)}R_{(v_1,v_2)}^{(2)}(x)+\tau_{(v_1,v_2+1)}d_{(v_1+1,v_2-1)}^{(2)}R_{(v_1+1,v_2-1)}^{(2)}(x),\nonumber\\
\tau_{(v_1-1,v_2)}d_{(v_1+1,v_2+1)}^{(1)}&R_{(v_1+1,v_2+1)}^{(1)}(x)=\tau_{(v_1+1,v_2)}d_{(v_1-1,v_2+1)}^{(1)}R_{(v_1-1,v_2+1)}^{(1)}(x)\label{sub2}\\
&-\tau_{(v_1,v_2+1)}d_{(v_1,v_2)}^{(1)}\tilde{R}_{(v_1,v_2)}^{(1)}(x)+\p_{t_1}\tau_{(v_1,v_2+1)}d_{(v_1,v_2)}^{(1)}R_{(v_1,v_2)}^{(1)}(x),\nonumber\\
\p_{t_1}\tau_{(v_1,v_2-1)}d_{(v_1,v_2)}^{(1)}&R_{(v_1,v_2)}^{(1)}(x)=\tau_{(v_1,v_2-1)}d_{(v_1,v_2)}^{(1)}\tilde{R}_{(v_1,v_2)}^{(1)}(x)\label{sub3}\\
&+\tau_{(v_1-1,v_2)}d_{(v_1+1,v_2-1)}^{(1)}R_{(v_1+1,v_2-1)}^{(1)}(x)-\tau_{(v_1+1,v_2)}d_{(v_1-1,v_2-1)}^{(1)}R_{(v_1-1,v_2-1)}^{(1)}(x),\nonumber\\
\p_{s_1}\tau_{(v_1-1,v_2)}d_{(v_1,v_2)}^{(2)}&R_{(v_1,v_2)}^{(2)}(x)=-\tau_{(v_1,v_2-1)}d_{(v_1-1,v_2+1)}^{(2)}R_{(v_1-1,v_2+1)}^{(2)}(x)\label{sub4}\\
&+\tau_{(v_1-1,v_2)}d_{(v_1,v_2)}^{(2)}\tilde{R}_{(v_1,v_2)}^{(2)}(x)+\tau_{(v_1,v_2+1)}d_{(v_1-1,v_2-1)}^{(2)}R_{(v_1-1,v_2-1)}^{(2)}(x).\nonumber
\end{align}
\end{subequations}
\end{proposition}

\begin{proof}
We verify the first equation by making the use of Pfaffian identity \eqref{a1b}, and the others could be similarly verified. 
Taking symbols 
\begin{align*}
a_1=v_1^{(1)},\quad a_2=v_2-1^{(2)},\quad a_3=v_2^{(2)},\quad a_4=x,\quad \star=\{0^{(1)},\cdots,v_1-1^{(1)},0^{(2)},\cdots,v_2-2^{(2)}\}
\end{align*} in \eqref{a1b}, we arrive at the desired formula from Pfaffian expressions \eqref{msop1}-\eqref{msop2} and realize that 
\begin{align*}
\p_{s_1}\tau_{(v_1+1,v_2)}=\pf(0^{(1)},\cdots,v_1^{(1)},0^{(2)},\cdots,v_2-2^{(2)},v_2^{(2)}).
\end{align*}
\end{proof}

Several simple integrable lattices could be obtained directly by using these relations. Expanding the linear form  \eqref{msop2}, we have
\begin{align*}
d_{(v_1,v_2)}^{(1)}R_{(v_1,v_2)}^{(1)}(x)&=(-1)^{v_1-1}\omega_1(x)\left(x^{v_1-1}\tau_{(v_1-1,v_2)}-x^{v_1-2}\p_{t_1}\tau_{(v_1-1,v_2)}+\cdots\right)\\
&+\omega_2(x)\left(
x^{v_2-1}\tau_{(v_1,v_2-1)}-x^{v_2-2}\p_{s_1}\tau_{(v_1,v_2-1)}+\cdots
\right).
\end{align*}
Moreover, if the equation \eqref{te} is taken into account, then the following equations 
\begin{align}\label{pfafftoda}
\begin{aligned}
&D_{t_1}\tau_{(v_1,v_2-1)}\cdot\tau_{(v_1,v_2+1)}=D_{s_1}\tau_{(v_1+1,v_2)}\cdot\tau_{(v_1-1,v_2)},\\
&D_{s_1}D_{t_1}\tau_{(v_1-1,v_2)}\cdot\tau_{(v_1-1,v_2)}=2\left(
\tau_{(v_1,v_2-1)}\tau_{(v_1-2,v_2+1)}-\tau_{(v_1,v_2+1)}\tau_{(v_1-2,v_2-1)}
\right)
\end{aligned}
\end{align}
are obtained by comparing the coefficients of $x^{v_1-2}\omega_1(x)$ and $x^{v_2-2}\omega_2(x)$ respectively. Here $D_t$ is the Hirota's bilinear operator defined by \cite{hirota04} 
\begin{align}\label{bo}
D_t^m D_x^n f(x,t)\cdot g(x,t)=\left.\frac{\p^m}{\p s^m}\frac{\p^n}{\p y^n} f(t+s,x+y)g(t-s,x-y)\right|_{s=0,y=0.}
\end{align}
Equations \eqref{pfafftoda} have appeared as a generalization of 2D Toda lattice, see for example 
\cite{hu05,takasaki09,santini04,gilson05,willox05}. Before we proceed to  further discussions about the recurrence, we demonstrate a reduction from MSOPs to skew orthogonal polynomials (SOPs).

\subsection{Reduction: from MSOPs to SOPs}
As was shown in \cite[Section 2.4]{takasaki09}, the hierarchy governing equations \eqref{pfafftoda} could be reduced to the DKP hierarchy from the perspective of fermionic representation. We reconfirm the fact in this part by performing reductions of MSOPs. 

By considering one index set $v=\{0,\cdots,2n\}$ and one weight function $\omega(x)$, we could define skew orthogonal polynomials $\{p_{2n}(x),p_{2n+1}(x)\}_{n\in\mathbb{N}}$ by the following skew orthogonal relation 
\begin{align*}
\langle p_{2n}(x),p_{2m}(x)\rangle=\langle p_{2n+1}(x),p_{2m+1}(x)\rangle=0,\quad \langle p_{2n}(x),p_{2m+1}(x)\rangle=\delta_{n,m},
\end{align*}
where $\langle\cdot,\cdot\rangle$ is a skew symmetric bilinear form on $\mathbb{R}[x]\times \mathbb{R}[y]\to\mathbb{R}$ and
\begin{align*}
\langle f(x),g(x)\rangle=\int_{\gamma\times\gamma} f(x)\s(x,y)g(y)\omega(x)\omega(y)dxdy,\quad \s(x,y)=-\s(y,x).
\end{align*}
This is a reductional version compared with  Proposition \ref{prop3.2}. Moreover, $\{p_{2n}(x),p_{2n+1}(x)\}_{n\in\mathbb{N}}$ are polynomials with Pfaffian expressions \cite[Thm. 3.1]{adler992}
\begin{align*}
p_{2n}(x)=d_n^{-1}\pf(0,\cdots,2n,x),\quad p_{2n+1}(x)=d_n^{-1}\pf(0,\cdots,2n-1,2n+1,x),
\end{align*}
where $d_n={(\tau_{2n}\tau_{2n+2})^{1/2}}$, $\tau_{2n}=\pf(0,\cdots,2n-1)$ and Pfaffian elements are given by 
\begin{align*}
\pf(i,j)=\langle x^i,y^j\rangle, \quad \pf(i,x)=x^i.
\end{align*}
By introducing the time flows $\mt=(t_1,t_2,\cdots)$ such that $\partial_{t_i}\omega(x;\mt)=x^i\omega(x;\mt)$, it was found that the skew orthogonal polynomials satisfy \cite[Thm. 3.1]{adler992}
\begin{align*}
(z+\p_{t_1})(d_np_{2n}(z))=d_np_{2n+1}(z).
\end{align*}
This equation coincides with equation \eqref{td} in multi-component case, and plays a role as spectral problem in integrable system theory. 

In literatures, the first study between SOPs and Pfaff lattice was carried out in \cite{adler992} from a view of  Lie algebra splitting. Later on, the correspondence was reformulated from different perspectives such as reductions from 2d-Toda theory \cite{adler02,li19}, Toda lattice and Pfaff lattice correspondence \cite{adler03}, symplectic matrices \cite{kodama10}, and so on.  
Therefore, it is natural to ask whether there is any local recurrence for SOPs which could be applied to derive integrable systems. Unfortunately, we couldn't find a compact relation between $p_{2n}(z)$ and $p_{2n+1}(z)$ as multi-component case in Proposition \ref{prop3.5}. By taking $\star=\{0,\cdots,2n-2\}$, $a_1=2n-1$, $a_2=2n$, $a_3=2n+1$ and $a_4=x$ in the identity \eqref{a1b} and using the equation
\begin{align*}
(\p_{t_2}+\p_{t_1}^2)\tau_{2n}=2\pf(0,1,\cdots,2n-2,2n+1),
\end{align*}
one has
\begin{align*}
\tau_{2n+2}d_{2n-2}p_{2n-2}(x)&=\frac{1}{2}(\p_{t_2}+\p_{t_1}^2)\tau_{2n}d_{2n}p_{2n}(x)-\p_{t_1}\tau_{2n}d_{2n}p_{2n+1}(z)\\&+\tau_{2n}\pf(0,\cdots,2n-2,2n,2n+1,z).
\end{align*}
This relation is non-compact since the last term could not be written in terms of skew orthogonal polynomials. However, due to the independency of function, integrable lattices could also be obtained by comparing the coefficients of monomials on both sides. The simplest equation arises when comparing the coefficients of $x^{2n-2}$, and one has
\begin{align*}
(D_1^4-4D_1D_3+3D_2^2)\tau_{2n}\cdot\tau_{2n}=24\tau_{2n-2}\tau_{2n+2}.
\end{align*}
This is the first member in the DKP hierarchy.

\section{Integrable lattice hierarchies from identities of MSOPs}\label{sec4}
In this part, we demonstrate that MSOPs could be expressed by 2-component Pfaffian $\tau$-functions $\{\tau_{(i,j)}(\mt,\ms)\}_{i,j\in\mathbb{N}}$ with $i+j\in 2\mathbb{N}$.  Since MSOPs are multi-component generalizations of SOPs, we call the corresponding integrable hierarchy as multiple-component Pfaff lattice hierarchy, especially a 2-component Pfaff lattice hierarchy in this paper. 

There are two different perspectives in deriving those integrable hierarchies, as mentioned in the introduction part. One is to express polynomials by $\tau$-functions. By substituting $\tau$-functions into recurrence relations, integrable hierarchy involving neighboring $\tau$-functions could be obtained. Another method is to make use of bilinear form and Cauchy transform. By using these methods, some famous integrable equations, such as the so-called Pfaff-Toda lattice and modified coupled KP equation are derived.
It is also shown that 2-component Pfaff lattice hierarchy derived from MSOPs is equivalent to Takasaki's Pfaff-Toda hierarchy.

\subsection{From recurrence relations \eqref{sub1}-\eqref{sub4} to integrable hierarchy}\label{sec4.1}
In this part, $\tau$-function expressions for the linear forms of MSOPs are given to characterize the corresponding integrable hierarchy. To this end, we first demonstrate an explicit connection between the linear forms of MSOPs and 2-component Pfaffian $\tau$-functions.
\begin{proposition}\label{prop4.1}
The linear forms $R_{(v_1,v_2)}^{(i)}(x;\mt,\ms)$ $(i=1,2)$ of multiple skew orthogonal polynomials  could be alternatively written by
\begin{align}\label{exp}
\begin{aligned}
d_{(v_1,v_2)}^{(i)}R_{(v_1,v_2)}^{(i)}(x;\mt,\ms)&=
(-1)^{v_1-1}\omega_1(x;\mt)x^{v_1-1}\tau_{(v_1-1,v_2)}(\mt-[x^{-1}],\ms)\\\
&+\omega_2(x;\ms)x^{v_2-1}\tau_{(v_1,v_2-1)}(\mt,\ms-[x^{-1}]),
\end{aligned}
\end{align}
where symbol $[\alpha]$ represents the Miwa variable 
\begin{align*}
[\alpha]=\left(
\alpha,\frac{\alpha^2}{2},\cdots,\frac{\alpha^n}{n},\cdots
\right).
\end{align*}
\end{proposition}
\begin{proof}
One could prove such a formula by column expansion to the moment matrix, and make use of Schur functions acting on moments; see e.g. \cite[prop 2.2]{adler09}. In our proof we adopt the method by directly acting Schur functions to $\tau$-functions.
Recall that the linear forms of MSOPs admit the Pfaffian expression
\begin{align*}
d_{(v_1,v_2)}^{(i)}R_{(v_1,v_2)}^{(i)}(x;\mt,\ms)=\pf(0^{(1)},\cdots,v_1-1^{(1)},0^{(2)},\cdots,v_2-1^{(2)},x).
\end{align*}
If we expand this formula from $x$, then we have
\begin{align}
\begin{aligned}\label{exp1}
d_{(v_1,v_2)}^{(i)}R_{(v_1,v_2)}^{(i)}(x;\mt,\ms)&=\omega_1(x;\mt)\sum_{i\in I_1} (-1)^i x^i\pf(I_1\backslash\{i\},I_2)\\&+\omega_2(x;\ms)\sum_{i\in I_2}(-1)^{v_2-1-i}x^i \pf(I_1,I_2\backslash\{i\}),
\end{aligned}
\end{align}
and index set $I_1=\{0^{(1)},\cdots,v_1-1^{(1)}\}$ and $I_2=\{0^{(2)},\cdots,v_2-1^{(2)}\}$.
Therefore, to demonstrate the equivalence between \eqref{exp} and \eqref{exp1}, one needs to verify the formula
\begin{align}\label{tf}
x^{v_1-1}\tau_{(v_1-1,v_2)}(\mt-[x^{-1}],\ms)=\sum_{i\in I_1}(-1)^{v_1-1-i}x^i\pf(I_1\backslash\{i\},I_2).
\end{align}
It is known that the left hand side in above formula could be written as
\begin{align*}
\tau_{(v_1-1,v_2)}(\mt-[x^{-1}],\ms)=e^{-\xi(\tilde{\p}_t,x^{-1})}\tau_{(v_1-1,v_2)}=\sum_{k\geq0}p_k(-\tilde{\p}_t)\tau_{(v_1-1,v_2)}x^{-k},
\end{align*}
where $\tilde{\p}_t=(\p_{{t}_1},\p_{{t}_2}/2,\cdots)$, $\xi(\mt,x)=\sum_{i=1}^\infty t_ix^i$ and $p_k$ are elementary symmetric  functions defined by
\begin{align}\label{esf}
e^{\xi(\mt,x)}=\sum_{k\geq 0}p_k(\mt)x^k.
\end{align}
Moreover, due to the Proposition \ref{propb1} in the appendix, we know the fact that
\begin{align*}
p_k(-\tilde{\p}_t)\tau_{(v_1-1,v_2)}=\Pf(0^{(1)},\cdots,\widehat{v_1-k^{(1)}},\cdots,v_1^{(1)},0^{(2)},\cdots,v_2^{(2)}),
\end{align*}
where $\hat{i}$ means that the index $i$ is missed, then equation \eqref{tf} holds.  
\end{proof}

\begin{remark}
According to the proof, we know that
\begin{align*}
d_{(v_1,v_2)}^{(i)}R_{(v_1,v_2)}^{(i)}(x;\mt,\ms)&=(-1)^{v_1-1}\omega_1(x;\mt)\sum_{\ell=0}^{v_1-1}\left(
p_\ell(-\tilde{\p}_t)\tau_{(v_1-1,v_2)}(\mt,\ms)\right)x^{v_1-1-\ell}\\
&+\omega_2(x;\ms)\sum_{\ell=0}^{v_2-1}\left(
p_\ell(-\tilde{\p}_s)\tau_{(v_1,v_2-1)}(\mt,\ms)
\right)x^{v_2-1-\ell}.
\end{align*}
\end{remark}

As a direct corollary, we have
\begin{coro}
 $\tilde{R}_{(v_1,v_2)}^{(i)}(x;\mt,\ms)$ $(i=1,2)$ could be expressed in terms of $\tau$-functions as
\begin{align*}
d_{(v_1,v_2)}^{(1)}&\tilde{R}_{(v_1,v_2)}^{(1)}(x;\mt,\ms)=
\p_{t_1}\left(d_{(v_1,v_2)}^{(1)}\tilde{R}_{(v_1,v_2)}^{(1)}(x;\mt,\ms)\right)\\
&=(-1)^{v_1-1}\omega_1(x;\mt)\sum_{\ell=0}^{v_1-1}\left(\p_{t_1}p_\ell(-\tilde{\p}_t)\tau_{(v_1-1,v_2)}(\mt,\ms)\right)x^{v_1-1-\ell}\\
&+(-1)^{v_1-1}\omega_1(x;\mt)\sum_{\ell=0}^{v_1-1}\left(p_\ell(-\tilde{\p}_t)\tau_{(v_1-1,v_2)}(\mt,\ms)\right)x^{v_1-\ell}\\&+\omega_2(x;\ms)
\sum_{\ell=0}^{v_2-1}\left(\p_{t_1}p_\ell(-\tilde{\p}_s)\tau_{(v_1,v_2-1)}(\mt,\ms)\right)x^{v_2-1-\ell},\\
d_{(v_1,v_2)}^{(2)}&\tilde{R}_{(v_1,v_2)}^{(2)}(x;\mt,\ms)=
\p_{s_1}\left(d_{(v_1,v_2)}^{(2)}\tilde{R}_{(v_1,v_2)}^{(2)}(x;\mt,\ms)\right)\\
&=(-1)^{v_1-1}\omega_1(x;\mt)\sum_{\ell=0}^{v_1-1}\left(\p_{s_1}p_\ell(-\tilde{\p}_t)\tau_{(v_1-1,v_2)}(\mt,\ms)\right)x^{v_1-1-\ell}\\
&+\omega_2(x;\ms)\sum_{\ell=0}^{v_2-1}\left(\p_{s_1}p_\ell(-\tilde{\p}_s)\tau_{(v_1,v_2-1)}(\mt,\ms)\right)x^{v_2-1-\ell}\\
&+\omega_2(x;\ms)\sum_{\ell=0}^{v_2-1}\left(p_\ell(-\tilde{\p}_s)\tau_{(v_1,v_2-1)}(\mt,\ms)\right)x^{v_2-\ell}.
\end{align*}
\end{coro}
By taking these expressions into \eqref{sub1}, and comparing the coefficients of $x^{v_1-j}\omega_1(x)$ and $x^{v_2-j}\omega_2(x)$ $(j=1,2,\cdots)$ respectively, we obtain
\begin{subequations}
\begin{align}
\tau_{(v_1,v_2-1)}p_j(-\tilde{\p}_t)\tau_{(v_1+1,v_2)}&=-\tau_{(v_1+1,v_2)}\p_{s_1}p_{j-1}(-\tilde{\p}_t)\tau_{(v_1-1,v_2)}\label{isub1}\\
&+\p_{s_1}\tau_{(v_1+1,v_2)}p_{j-1}(-\tilde{\p}_t)\tau_{(v_1-1,v_2)}+\tau_{(v_1,v_2+1)}p_j(-\tilde{\p}_t)\tau_{(v_1,v_2-1)},\nonumber\\
\tau_{(v_1,v_2-1)}p_j(-\tilde{\p}_s)\tau_{(v_1+1,v_2)}&=\tau_{(v_1+1,v_2)}\left(
\p_{s_1}p_{j-1}(-\tilde{\p}_s)+p_j(-\tilde{\p}_s)
\right)\tau_{(v_1,v_2-1)}\label{isub2}\\
&-\p_{s_1}\tau_{(v_1+1,v_2)}p_{j-1}(-\tilde{\p}_s)\tau_{(v_1,v_2-1)}+\tau_{(v_1,v_2+1)}p_{j-2}(-\tilde{\p}_s)\tau_{(v_1+1,v_2-2)}.\nonumber
\end{align}
\end{subequations}
Moreover, we read from \eqref{sub3} that
\begin{subequations}
\begin{align}
\p_{t_1}\tau_{(v_1,v_2-1)}p_{j-1}(-\tilde{\p}_t)\tau_{(v_1-1,v_2)}&=\tau_{(v_1,v_2-1)}\left(
\p_{t_1}p_{j-1}(-\tilde{\p}_t)+p_j(-\tilde{\p}_t)\right)\tau_{(v_1-1,v_2)}\label{isub3}\\
&-\tau_{(v_1-1,v_2)}p_j(-\tilde{\p}_t)\tau_{(v_1,v_2-1)}+\tau_{(v_1+1,v_2)}p_{j-2}(-\tilde{\p}_t)\tau_{(v_1-2,v_2-1)},\nonumber\\
\p_{t_1}\tau_{(v_1,v_2-1)}p_{j-1}(-\tilde{\p}_s)\tau_{(v_1,v_2-1)}&=\tau_{(v_1,v_2-1)}\p_{t_1}p_{j-1}(-\tilde{\p}_s)\tau_{(v_1,v_2-1)}\label{isub4}\\
&+\tau_{(v_1-1,v_2)}p_{j-2}(-\tilde{\p}_s)\tau_{(v_1+1,v_2-2)}-\tau_{(v_1+1,v_2)}p_{j-2}(-\tilde{\p}_s)\tau_{(v_1-1,v_2-2)}.\nonumber
\end{align}
\end{subequations}
It should be remarked that integrable hierarchies \eqref{isub1}-\eqref{isub2} and \eqref{isub3}-\eqref{isub4} are the same if one interchanges $v_1$ with $v_2$ and $\p_t$ with $\p_s$. Moreover, integrable hierarchies derived from \eqref{sub2} and \eqref{sub4} are the same with \eqref{sub1} and \eqref{sub3}. Therefore, it is reasonable to regard \eqref{isub1}-\eqref{isub4} as a 2-component Pfaff lattice hierarchy with neighboring lattices. 

There are some integrable lattices obtained from those hierarchies. The first equation of Pfaff-Toda lattice in \eqref{pfafftoda} could be obtained from \eqref{isub1} by taking $j=1$, and the second one could be obtained from \eqref{isub4} by taking $j=2$. Besides, one could obtain another nontrivial simple example in \eqref{isub2} when $j=2$, which reads
\begin{align}\label{mckp}
(D_{s_2}+D_{s_1}^2)\tau_{(v_1,v_2-1)}\cdot\tau_{(v_1+1,v_2)}=2\tau_{(v_1,v_2+1)}\tau_{(v_1+1,v_2-2)}.
\end{align}
This is the bilinear form of the so-called modified coupled KP equation, which plays an important role in the study of commutativity of Pfaffianization and B\"acklund transformation \cite{hu052}.
\subsection{Bilinear identities: from bilinear form to Cauchy transforms}\label{sec4.2}
In last subsection, we derived a 2-component Pfaff lattice hierarchy by directly using the recurrence relations of MSOPs, which involve neighboring $\tau$-functions.
In this part, we find another approach to deduce more general integrable lattice hierarchies from the perspective of Cauchy transforms. To this end, we introduce a Cauchy transform with respect to a non-degenerate bilinear form.
\begin{proposition}\label{prop5.1}
Given a non-degenerate bilinear form $\langle\cdot,\cdot\rangle:\,\mathbb{R}[x]\times\mathbb{R}[y]\to\mathbb{R}$ and an analytic weight function $\psi(x)$, then for an integrable function $g(x)$, a Cauchy transform of $g(x)$ with respect to the bilinear form is defined by
\begin{align*}
\mc_\psi g(z)=\left\langle \frac{\psi(x)}{x-z},g(y)\right\rangle.
\end{align*}
Moreover, for any analytic function $f(x)$, one has
\begin{align*}
\langle f(x)\psi(x),g(y)\rangle=\frac{1}{2\pi i}\oint_{C_\infty}f(z)\mc_\psi g(z)dz,
\end{align*}
where $C_\infty$ is a circle around the infinity.
\end{proposition}
\begin{proof}
By assuming that $f(z)$ is analytic, we have the expansion $f(z)=\sum_{i=0}^\infty f_iz^i$, and thus
\begin{align*}
\frac{1}{2\pi i}\oint_{C_\infty} f(z)\mc_\psi g(z)dz&=\frac{1}{2\pi i}\oint_{C_\infty}\sum_{i=0}^\infty f_iz^i\sum_{j=0}^\infty \frac{1}{z^{j+1}}\langle x^j\psi(x),g(y)\rangle dz\\
&=\sum_{i=0}^\infty f_i\langle x^i\psi(x),g(y)\rangle=\langle f(x)\psi(x),g(y)\rangle.
\end{align*}
\end{proof}
Therefore, by taking $\langle\cdot,\cdot\rangle$ as a skew symmetric bilinear form, i.e.
\begin{align*}
\langle f(x),g(y)\rangle=\int_{\gamma\times\gamma}f(x)\s(x,y)g(y)dxdy,\quad \s(x,y)=-\s(y,x),
\end{align*}
one could define a corresponding Cauchy transform 
\begin{align}\label{ct}
\mc_\psi g(z)=\int_{\gamma\times \gamma}\frac{\psi(z)}{x-z}\s(x,y)g(y)dxdy.
\end{align}
Moreover, we have the following statement which is about the Cauchy transform of MSOPs with skew inner product.
\begin{proposition}\label{prop5.2}
If $R_{(v_1,v_2)}^{(i)}(x;\mt,\ms)$ $(i=1,2)$ are linear forms of multiple skew orthogonal polynomials defined in Proposition \ref{prop3.2} with weights $\omega_1(x;\mt)$ and $\omega_2(x;\ms)$, then we have 
\begin{align*}
&\mc_{\omega_1} \left(
d_{(v_1,v_2)}^{(i)}R_{(v_1,v_2)}^{(i)}
\right)=(-1)^{v_1}z^{-(v_1+1)}\tau_{(v_1+1,v_2)}(\mt+[z^{-1}],\ms),\\
&\mc_{\omega_2} \left(
d_{(v_1,v_2)}^{(i)}R_{(v_1,v_2)}^{(i)}
\right)=z^{-(v_2+1)}\tau_{(v_1,v_2+1)}(\mt,\ms+[z^{-1}]).
\end{align*}
\end{proposition}
\begin{proof}
We prove the first equation, and the second one could be similarly verified.
By using \eqref{msop1} and \eqref{ct}, we have
\begin{align*}
\mc_{\omega_1} &\left(
d_{(v_1,v_2)}^{(i)}R_{(v_1,v_2)}^{(i)}
\right)=\int_{\gamma\times\gamma}\frac{\omega_1(x;\mt)}{x-z}\s(x,y)\Pf(0^{(1)},\cdots,v_1-1^{(1)},0^{(2)},\cdots,v_2-1^{(2)},y)dxdy\\
&=-\sum_{i=0}^\infty z^{-(i+1)}\int_{\gamma\times\gamma}x^i\omega_1(x;\mt)\s(x,y)\Pf(0^{(1)},\cdots,v_1-1^{(1)},0^{(2)},\cdots,v_2-1^{(2)},y)dxdy\\
&=-\sum_{i=0}^\infty z^{-(i+1)}\left\langle \pf(i^{(1)},x),\pf(0^{(1)},\cdots,v_1-1^{(1)},0^{(2)},\cdots,v_2-1^{(2)},y)\right\rangle.
\end{align*}
Then from the skew orthogonality, when $i\leq v_1-1$ the above skew inner product is equal to zero. Therefore, the above formula is equal to 
\begin{align*}
&-\sum_{i=v_1}^\infty (-1)^{v_2}z^{-(i+1)}\Pf(0^{(1)},\cdots,v_1-1^{(1)},i^{(1)},0^{(2)},\cdots,v_2-1^{(2)})\\
&=(-1)^{v_1}z^{-(v_1+1)}\sum_{i=0}^\infty z^{-i}p_{i}(\tilde{\p}_t)\tau_{(v_1+1,v_2)},
\end{align*}
which is the expansion of the desired formula.
\end{proof}
In the followings, we show how to derive integrable hierarchy and bilinear identities by Cauchy transforms.
\begin{proposition}
Two-component $\tau$-functions $\{\tau_{(i,j)}(\mt,\ms)\}_{i,j\in\mathbb{N}}$ with $i+j\in2\mathbb{N}$ satisfy bilinear identity
\begin{align}\label{bi}
\begin{aligned}
&(-1)^{u_1+v_1}\oint_{C_\infty} e^{\xi(t-t',z)}z^{v_1-u_1-2}\tau_{(v_1-1,v_2)}(\mt-[z^{-1}],\ms)\tau_{(u_1+1,u_2)}(\mt'+[z^{-1}],\ms')dz\\
&\qquad\qquad+(-1)^{u_1+v_1}\oint_{C_\infty}e^{\xi(t'-t,z)}z^{u_1-v_1-2}\tau_{(v_1+1,v_2)}(\mt+[z^{-1}],\ms)\tau_{(u_1-1,u_2)}(\mt'-[z^{-1}],\ms')dz\\
&=\oint_{C_\infty} e^{\xi(s-s',z)}z^{v_2-u_2-2}\tau_{(v_1,v_2-1)}(\mt,\ms-[z^{-1}])\tau_{(u_1,u_2+1)}(\mt',\ms'+[z^{-1}])dz\\
&\qquad\qquad+\oint_{C_\infty}e^{\xi(s'-s,z)}z^{u_2-v_2-2}\tau_{(v_1,v_2+1)}(\mt,\ms+[z^{-1}])\tau_{(u_1,u_2-1)}(\mt',\ms'-[z^{-1}])dz,
\end{aligned}
\end{align}
which is valid for arbitrary $t,t',s,s'\in\mathbb{C}$.
\end{proposition}
\begin{proof}
Since $\langle\cdot,\cdot\rangle$ is a skew inner product, we know that
\begin{align*}
\langle R_{(v_1,v_2)}^{(1)}(x;\mt,\ms),R^{(1)}_{(u_1,u_2)}(y;\mt',\ms')\rangle=-\langle R_{(u_1,u_2)}^{(1)}(x;\mt',\ms'),R_{(v_1,v_2)}^{(1)}(y;\mt,\ms)\rangle,
\end{align*}
holds true for arbitrary $t,t',s,s'\in\mathbb{C}$ and $|\vec{u}|,\,|\vec{v}|\in2\mathbb{N}+1$. By multiplying $d_{(u_1,u_2)}^{(1)}d_{(v_1,v_2)}^{(1)}$ on both sides and expanding the linear forms of MSOPs in terms of $\tau$-function according to Prop. \ref{prop4.1}, we have
\begin{align*}
&(-1)^{v_1-1}\left\langle x^{v_1-1}\tau_{(v_1-1,v_2)}(\mt-[x^{-1}],\ms)e^{\xi(t,x)}\omega_1(x),d_{(u_1,u_2)}^{(1)}R_{(u_1,u_2)}(y;\mt',\ms')\right\rangle\\
&\quad +\left\langle x^{v_2-1}\tau_{(v_1,v_2-1)}(\mt,\ms-[x^{-1}])e^{\xi(s,x)}\omega_2(x),d_{(u_1,u_2)}^{(1)}R_{(u_1,u_2)}(y;\mt',\ms')\right\rangle\\
& =(-1)^{u_1}\left\langle x^{u_1-1}\tau_{(u_1-1,u_2)}(\mt'-[x^{-1}],\ms')e^{\xi(t',x)}\omega_1(x),d_{(v_1,v_2)}^{(1)}R_{(v_1,v_2)}(y;\mt,\ms)\right\rangle\\
&\quad -\left\langle x^{u_2-1}\tau_{(u_1,u_2-1)}(\mt',\ms'-[x^{-1}])e^{\xi(s',x)}\omega_2(x),d_{(v_1,v_2)}^{(1)}R_{(v_1,v_2)}(y;\mt,\ms)\right\rangle.
\end{align*}
Then by realizing that $\omega_1(x;\mt)=e^{\xi(x;\mt-\mt')}\omega_1(x;\mt')$ and according to Proposition \ref{prop5.1}, we have
\begin{align*}
&(-1)^{v_1-1}\frac{1}{2\pi i}\oint_{C_\infty} e^{\xi(t-t',z)}z^{v_1-1}\tau_{(v_1-1,v_2)}(\mt-[z^{-1}],\ms)\mc_{\omega_1}\left(
d_{(u_1,u_2)}^{(1)}R_{(u_1,u_2)}^{(1)}
\right)(z;\mt',\ms')dz\\
&\quad +\frac{1}{2\pi i}\oint_{C_\infty}e^{\xi(s-s',z)}z^{v_2-1}\tau_{(v_1,v_2-1)}(\mt,\ms-[z^{-1}])\mc_{\omega_2}\left(
d_{(u_1,u_2)}^{(1)}R_{(u_1,u_2)}^{(1)}
\right)(z;\mt',\ms')dz\\
&=(-1)^{u_1}\frac{1}{2\pi i}\oint_{C_\infty} e^{\xi(t'-t,z)}z^{u_1-1}\tau_{(u_1-1,u_2)}(\mt'-[z^{-1}],\ms')\mc_{\omega_1}\left(
d_{(v_1,v_2)}^{(1)}R_{(v_1,v_2)}^{(1)}
\right)(z;\mt,\ms)dz\\
&\quad-\frac{1}{2\pi i}\oint_{C_\infty} e^{\xi(s'-s,z)}z^{u_2-1}\tau_{(u_1,u_2-1)}(\mt',\ms'-[z^{-1}])\mc_{\omega_2}\left(
d_{(v_1,v_2)}^{(1)}R_{(v_1,v_2)}^{(1)}
\right)(z;\mt,\ms)dz,
\end{align*}
By substituting Cauchy transforms in Prop. \ref{prop5.2} into above formula, we complete the proof.
\end{proof}

\begin{remark}
Bilinear identity \eqref{bi} should coincide with \cite[eq. (2.1)]{takasaki09} if one changes $z$ to $z^{-1}$ and transforms the contour around the infinity into a circle around zero.
\end{remark}

If we take the variable transformations
\begin{align*}
\mt\mapsto \mt-\alpha,\,\mt'\mapsto \mt+\alpha,\, \ms\mapsto \ms-\beta, \,\ms'\mapsto \ms+\beta,
\end{align*}
and realizes
\begin{align*}
\tau_{(m,n)}(\mt+\alpha+[z^{-1}],\ms+\beta)\tau_{(u,v)}(\mt-\alpha-[z^{-1}],\ms-\beta)=e^{\sum_{i=1}^\infty \alpha_iD_{t_i}+\beta_iD_{s_i}+\xi(\tilde{D}_t,z^{-1})}\tau_{(m,n)}\tau_{(u,v)}
\end{align*}
for arbitrary $m+n,u+v\in2\mathbb{N}+1$,
then the identity \eqref{bi} becomes 
\begin{align*}
&(-1)^{u_1+v_1}
\oint_{C_\infty} e^{-2\xi(\alpha,z)}z^{v_1-u_1-2}e^{\sum_{i=1}^\infty (\alpha_iD_{t_i}+\beta_iD_{s_i})-\xi(\tilde{D}_t,z^{-1})}\tau_{(u_1+1,u_2)}\cdot\tau_{(v_1-1,v_2)}dz\\
&\qquad+(-1)^{u_1+v_1}\oint_{C_\infty} e^{2\xi(\alpha,z)}z^{u_1-v_1-2}e^{\sum_{i=1}^\infty (\alpha_iD_{t_i}+\beta_iD_{s_i})+\xi(\tilde{D}_t,z^{-1})}\tau_{(u_1-1,u_2)}\cdot\tau_{(v_1+1,v_2)}\\
&=\oint_{C_\infty} e^{-2\xi(\beta,z)}z^{v_2-u_2-2}e^{\sum_{i=1}^\infty (\alpha_iD_{t_i}+\beta_iD_{s_i})+\xi(\tilde{D}_s,z^{-1})}\tau_{(u_1,u_2+1)}\cdot\tau_{(v_1,v_2-1)}dz\\
&\qquad+\oint_{C_\infty} e^{2\xi(\beta,z)}z^{u_2-v_2-2}e^{\sum_{i=1}^\infty (\alpha_iD_{t_i}+\beta_iD_{s_i})-\xi(\tilde{D}_s,z^{-1})}\tau_{(u_1,u_2-1)}\cdot\tau_{(v_1,v_2+1)}dz.
\end{align*}
Therefore, according to the residue theorem, it is equivalent to 
\begin{align*}
&(-1)^{u_1+v_1}\sum_{n=0}^\infty e_n(-2\alpha)e_{n+v_1-u_1-1}(\tilde{D}_t)e^{\sum_{i=1}^\infty \alpha_iD_{t_i}+\beta_iD_{s_i}}\tau_{(u_1+1,u_2)}\cdot \tau_{(v_1-1,v_2)}\\
&\qquad +(-1)^{u_1+v_1}\sum_{n=0}^\infty e_n(2\alpha)e_{n+u_1-v_1-1}(-\tilde{D}_t)e^{\sum_{i=1}^\infty \alpha_iD_{t_i}+\beta_iD_{s_i}}\tau_{(u_1-1,u_2)}\cdot \tau_{(v_1+1,v_2)}\\
&=\sum_{n=0}^\infty e_n(-2\beta)e_{n+v_2-u_2-1}(\tilde{D}_s)e^{\sum_{i=1}^\infty \alpha_iD_{t_i}+\beta_iD_{s_i}}\tau_{(u_1,u_2+1)}\cdot \tau_{(v_1,v_2-1)}\\
&\qquad +\sum_{n=0}^\infty e_n(2\beta)e_{n+u_2-v_2-1}(-\tilde{D}_s)e^{\sum_{i=1}^\infty \alpha_iD_{t_i}+\beta_iD_{s_i}}\tau_{(u_1,u_2-1)}\cdot \tau_{(v_1,v_2+1)},
\end{align*}
where $\{e_k\}_{k\geq0}$ are elementary symmetric functions defined by \eqref{esf} and $D_t,\, D_s$ are bilinear operators given by \eqref{bo}.

Therefore, by comparing with the coefficients of $\alpha_1^m\beta_1^n$ for $m,n\geq 0$, we obtain the following integrable lattice hierarchies
\begin{align}\label{lattice}
\begin{aligned}
&(-1)^{u_1+v_1}
\frac{1}{n!}D_{s_1}^n \left(\sum_{k+l=m,k,l\geq0}\frac{(-2)^k}{l!} p_{k+v_1-u_1-1}(\tilde{D}_t)D_{t_1}^l
\right)\tau_{(u_1+1,u_2)}\cdot\tau_{(v_1-1,v_2)}\\
&\qquad+(-1)^{u_1+v_1}
\frac{1}{n!}D_{s_1}^n \left(\sum_{k+l=m,k,l\geq0}\frac{2^k}{l!} p_{k+u_1-v_1-1}(-\tilde{D}_t)D_{t_1}^l
\right)\tau_{(u_1-1,u_2)}\cdot\tau_{(v_1+1,v_2)}\\
&=\frac{1}{m!}D_{t_1}^m\left(
\sum_{k+l=n,k,l\geq0} \frac{(-2)^k}{l!}p_{k+v_2-u_2-1}(\tilde{D}_s)D_{s_1}^l
\right)\tau_{(u_1,u_2+1)}\cdot\tau_{(v_1,v_2-1)}\\
&\qquad+\frac{1}{m!}D_{t_1}^m\left(
\sum_{k+l=n,k,l\geq0} \frac{2^k}{l!}p_{k+u_2-v_2-1}(-\tilde{D}_s)D_{s_1}^l
\right)\tau_{(u_1,u_2-1)}\cdot\tau_{(v_1,v_2+1)}.
\end{aligned}
\end{align}
The first equation in \eqref{pfafftoda} is re-derived if $(u_1,u_2)=(v_1,v_2)$ and $(m,n)=(1,1)$, and the second equation in \eqref{pfafftoda} is re-derived if $(u_1,u_2)=(v_1-2,v_2)$ and $(m,n)=(0,1)$.

To conclude, we can give molecule solutions to the 2-component Pfaff lattice hierarchy. 
\begin{proposition}
The 2-component Pfaff lattice hierarchy \eqref{bi} admit the following molecule solutions
\begin{align*}
\tau_{(v_1,v_2)}=\pf(0^{(1)},\cdots,v_1-1^{(1)},0^{(2)},\cdots,v_2-1^{(2)}),\quad v_1,v_2\in\mathbb{N},\, v_1+v_2\in 2\mathbb{N}
\end{align*}
with $\tau_{(0,0)}=1$. Moreover, those Pfaffian elements satisfy following time evolutions
\begin{align*}
&\p_{t_n}\pf(i^{(1)},j^{(1)})=\pf(i+n^{(1)},j^{(1)})+\pf(i^{(1)},j+n^{(1)}),&\p_{t_n}\pf(i^{(1)},j^{(2)})=\pf(i+n^{(1)},j^{(2)}),\\
&\p_{s_n}\pf(i^{(2)},j^{(2)})=\pf(i+n^{(2)},j^{(2)})+\pf(i^{(2)},j+n^{(2)}),&\p_{s_n}\pf(i^{(1)},j^{(2)})=\pf(i^{(1)},j+n^{(2)}),\\
&\p_{t_n}\pf(i^{(2)},j^{(2)})=\p_{s_n}\pf(i^{(1)},j^{(1)})=0.
\end{align*}
\end{proposition}

\section{Concluding remarks}
In this paper, we develop ideas for how to properly define multiple skew orthogonal polynomials. This  concept should be appealing, as multiple orthogonal polynomials have been widely investigated in the fields of random matrices and integrable systems. As an application, we considered appropriate time deformations on multiple skew orthogonal polynomials, which were turned out to have tight connections with Pfaff-Toda hierarchy considered earlier by Takasaki. 
In our paper, we called the corresponding integrable hierarchy as 2-component Pfaff lattice hierarchy because they could be viewed from the perspective of multiple skew orthogonal polynomials. As mentioned in Takasaki's paper \cite{takasaki09}, Pfaff lattice hierarchy and multi-component Pfaff lattice hierarchy have many common properties. However, multiple skew orthogonal polynomials have compact recurrence relations shown in \eqref{sub1}-\eqref{sub4}, which play important roles in the formulation of  spectral problems for 2-component Pfaff lattice hierarchy. 

There are still interesting problems to continue. One is to seek for proper applications into random matrix theory. Both the Gaussian and chiral unitary models with a source are examples of determinantal point processes. In random matrix theory, Pfaffian point processes also arise naturally, we expect to find a random matrix model characterized by those multiple skew orthogonal polynomials. Besides, there are several 2-component BKP hierarchies \cite{shiota89,kac97} and whether their solutions are related to those multiple skew orthogonal polynomials is worthy studying.

\section*{Acknowledgement}
The authors thank Prof. Peter Forrester for his useful comments. 

\section*{Disclosure statement}
There is no any potential conflict of interest.

\section*{Funding}
S. Li was partially supported by the National Natural Science Foundation of China (Grant no. 12101432, 12175155), and G. Yu was supported by National Natural Science Foundation of China (Grant no. 11871336).

\begin{appendix}

\section{Pfaffian identitis}

There are two different kinds of Pfaffian identities, c.f. \cite[eq. 2.95' $\&$ 2.96']{hirota04} 
\begin{subequations}
\begin{align}
&\pf(\ast,a_1,a_2,a_3,a_4)\pf(\ast)=\pf(\ast,a_1,a_2)\pf(\ast,a_3,a_4)\nonumber\\
&\qquad\qquad-\pf(\ast,a_1,a_3)\pf(\ast,a_2,a_4)+\pf(\ast,a_1,a_4)\pf(\ast,a_2,a_3),\label{a1a}\\
&\pf(\star,a_1,a_2,a_3)\pf(\star,a_4)=\pf(\star,a_2,a_3,a_4)\pf(\star,a_1)\nonumber\\
&\qquad\qquad-\pf(\star,a_1,a_3,a_4)\pf(\star,a_2)+\pf(\star,a_1,a_2,a_4)\pf(\star,a_3),\label{a1b}
\end{align}
\end{subequations}
where $\ast$ and $\star$ are sets of even and odd-number symbols respectively.

\section{Derivative formulas for Wronskian type Pfaffians}\label{apb}
Wronskian-type Pfaffians are well investigated in soliton theory due to its wide applications in coupled KP theory. In \cite[Sec. 3.4]{hirota04}, Pfaffian element $\pf(i,j)$ satisfying the differential rules with respect to the variables $\mt=(t_1,t_2,\cdots)$ by
\begin{align}\label{d1}
\p_{t_n}\pf(i,j)=\pf(i+n,j)+\pf(i,j+n)
\end{align}
was called Wronskian type Pfaffians. For more details about Wronskian type Pfaffian and its discrete counterparts, please refer to \cite{ohta04}. It was shown that if Pfaffian elements satisfy \eqref{d1}, then
\begin{align*}
\p_{t_n}\Pf(i_0,i_1,\cdots,i_{2N-1})=\sum_{k=0}^{2N-1}\Pf(i_0,\cdots,i_k+n,\cdots,i_{2N-1}).
\end{align*}
This was proved by induction. 

In this paper, we need to introduce 2-component Pfaffian $\tau$-functions, 
indexed by $\mathcal{I}=\{i_0,\cdots,i_n\}$ and $\mathcal{J}=\{j_0,\cdots,j_m\}$ with $n,m\in\mathbb{N}$ and $n+m\in2\mathbb{N}$. Pfaffian elements in this case should satisfy a 2-component Wronskian type generalization (c.f. Prop \ref{prop3.3})
\begin{align*}
&\p_{t_k}\pf(i_\alpha,i_\beta)=\pf(i_\alpha+k,i_\beta)+(i_\alpha,i_\beta+k),&&\p_{t_k}\pf(i_\alpha,j_\beta)=\pf(i_\alpha+k,j_\beta),&&&\p_{t_k}\pf(j_\alpha,j_\beta)=0,\\
&\p_{s_k}\pf(j_\alpha,j_\beta)=\pf(j_\alpha+k,j_\beta)+(j_\alpha,j_\beta+k), &&\p_{s_k}\pf(i_\alpha,j_\beta)=\pf(i_\alpha,j_\beta+k),&&&\p_{s_k}\pf(i_\alpha,i_\beta)=0,
\end{align*}
then we have the following proposition.
\begin{proposition}\label{propb1}
If Pfaffian elements satisfy the above derivative relations, then one has
\begin{align*}
&\p_{t_k}\pf(i_0,\cdots,i_n,j_0,\cdots,j_m)=\sum_{\alpha=0}^n \pf(i_0,\cdots,i_\alpha+k,\cdots,i_n,j_0,\cdots,j_m),\\
&\p_{s_k}\pf(i_0,\cdots,i_n,j_0,\cdots,j_m)=\sum_{\alpha=0}^m \pf(i_0,\cdots,i_n,j_0,\cdots,j_\alpha+k,\cdots,j_m).
\end{align*}
\end{proposition}
\begin{proof}
Here we only prove the first equation by using induction, the second one can be similarly proved. Noting that
\begin{align*}
\p_{t_k}\pf(i_0,\cdots,i_n,j_0,\cdots,j_m)&=\p_{t_k}\left(
\sum_{i_l\in\mathcal{I}}(-1)^{l-1}\pf(i_0,i_l)\pf(i_1,\cdots,\hat{i}_l,\cdots,i_n,j_0,\cdots,j_m)\right.\\
&\qquad\left.+\sum_{j_l\in \mathcal{J}}(-1)^{n+l}\pf(i_0,j_l)\pf(i_1,\cdots,i_n,j_0,\cdots,\hat{j}_l,\cdots,j_m)\right)
\end{align*}
where the first part is equal to
\begin{subequations}
\begin{align}
&\sum_{i_l\in\mathcal{I}}(-1)^{l-1}\pf(i_0+k,i_l)\pf(i_1,\cdots,\hat{i}_l,\cdots,i_n,j_0,\cdots,j_m)\label{part1}\\
&+\sum_{i_l\in\mathcal{I}}(-1)^{l-1}\pf(i_0,i_l+k)\pf(i_1,\cdots,\hat{i}_l,\cdots,i_n,j_0,\cdots,j_m)\label{part2}\\
&+\sum_{i_l\in\mathcal{I}}(-1)^{l-1}\pf(i_0,i_l)\sum_{\alpha\ne l}\pf(i_1,\cdots,i_\alpha+k,\cdots,\hat{i}_l,\cdots,i_n,j_0,\cdots,j_m),\label{part3}
\end{align}
\end{subequations}
while the derivative of the second part is equal to 
\begin{subequations}
\begin{align}
&\sum_{j_l\in\mathcal{J}}(-1)^{n+l}\pf(i_0+k,j_l)\pf(i_1,\cdots,i_n,j_0,\cdots,\hat{j}_l,\cdots,j_m)\label{part4}\\
&+\sum_{j_l\in\mathcal{J}}(-1)^{n+l}\pf(i_0,j_l)\sum_{\alpha=1}^n\pf(i_1,\cdots,i_\alpha+k,\cdots,i_n,j_0,\cdots,\hat{j}_l,\cdots,j_m).\label{part5}
\end{align}
\end{subequations}
Therefore, by summing \eqref{part1} and \eqref{part4} up, one obtains
\begin{align*}
\pf(i_0+k,i_1,\cdots,i_n,j_0,\cdots,j_m).
\end{align*}
The summation of rest three equations is equal to 
\begin{align*}
\sum_{\alpha=1}^n \pf(i_0,\cdots,i_\alpha+k,\cdots,i_n,j_0,\cdots,j_m),
\end{align*}
and our proof is complete.
\end{proof}

\end{appendix}

\end{document}